%% file: main.tex
\documentclass[a4paper,UKenglish,pdfa,thm-restate,numberwithinsect]{article}
\usepackage{amsmath, amsthm, amscd, amsfonts, amssymb, graphicx, color}
\usepackage[bookmarksnumbered, colorlinks, plainpages]{hyperref}
\usepackage[
    backend=biber,
    style=alphabetic,
    sortlocale=de_DE,
    natbib=true,
    url=false, 
    doi=true,
    eprint=false
]{biblatex}

\usepackage{mdframed}
\usepackage{comment}
\usepackage[toc,page]{appendix}

\usepackage{mathtools}
\usepackage{comment}

\DeclarePairedDelimiterX\braket[2]{\langle}{\rangle}{#1\,\delimsize\vert\,\mathopen{}#2}
\newtheorem{theorem}{Theorem}[section]
\newtheorem{remark}{Remark}[section]

\newtheorem{fact}[theorem]{Fact}
\usepackage{tikz}
\usetikzlibrary{shapes.multipart}

\newtheorem{lemma}[theorem]{Lemma}

\newtheorem{corollary}[theorem]{Corollary}
\theoremstyle{definition}
\newtheorem{definition}[theorem]{Definition}

\newtheorem{claim}[theorem]{Claim}

\newtheorem{observation}[theorem]{Observation}

\theoremstyle{remark}
\numberwithin{equation}{section}

\title{Exponential Lower Bounds on the Size of ResLin Proofs of Nearly Quadratic Depth}

\author{
Sreejata Kishor Bhattacharya
\thanks{TIFR, Mumbai. Email: {\tt sreejata.bhattacharya@tifr.res.in}. Supported by the Department of Atomic Energy, Govmt. of India, under project \#RTI4001 and by a Google PhD Fellowship.}
\and
Arkadev Chattopadhyay
\thanks{TIFR, Mumbai. Email: {\tt arkadev.c@tifr.res.ints}. Supported by the Department of Atomic Energy, Govmt. of India, under project \#RTI4001 and by a Google India Faculty Award.}}

\begin{document}
\maketitle
\setcounter{page}{1}

\noindent {\small }\hfill     {\small  }\\
{\small }\hfill  {\small }






\newcommand{\cube}{ \{ \pm 1\} }
\newcommand{ \infl}{ \mathsf{Inf}}
\newcommand{ \prob} {\mathsf{Pr}}
\newcommand{ \var} {\mathsf{Var}}
\newcommand{\E}{\mathsf{E}}
\newcommand {\poly}{\mathsf{poly}}
\newcommand{\Clh}{\tilde{\text{Cl}}}
\newcommand{\Cl}{\text{Cl}}
\newcommand{\F}{\mathbb{F}_2}
\newcommand{\IP}{\text{IP}}
\newcommand{\reslin}{\text{Res}(\oplus)}

\newcommand{\DTFooler}{DTFooling}

\newcommand{\ac}[1]{{\color{red} [{\bf AC:} #1]}}
\newcommand{\skb}[1]{{\color{blue} [{\bf SB:} #1]}}

\begin{abstract}

Itsykson and Sokolov \cite{IS14} identified resolution over parities, denoted by $\text{Res}(\oplus)$, as a natural and simple fragment of $\text{AC}^0[2]$-Frege for which no super-polynomial lower bounds on size of proofs are known. Building on a recent line of work (\cite{EGI24}, \cite{BCD24}, \cite{AI25}), Efremenko and Itsykson \cite{EI25} proved lower bounds of the form $\text{exp}(N^{\Omega(1)})$, on the size of $\text{Res}(\oplus)$ proofs whose depth is upper bounded by $O(N\log N)$, where $N$ is the number of variables of the unsatisfiable CNF formula. The hard formula they used was Tseitin on an appropriately expanding graph, lifted by a $2$-stifling gadget. They posed the natural problem of proving super-polynomial lower bounds on the size of proofs that are $\Omega(N^{1+\epsilon})$ deep, for any constant $\epsilon > 0$.

We prove the first such lower bounds. In fact, we show that $\text{Res}(\oplus)$ refutations of Tseitin formulas on constant-degree expanders on $m$ vertices, lifted with Inner-Product gadget of size $O(\log m)$, must have size $\text{exp}(\tilde{\Omega}(N^{\epsilon}))$, as long as the depth of the $\text{Res}(\oplus)$ proofs are $O(N^{2-\epsilon})$, for every $\epsilon > 0$. Here $N=\Theta(m\log m)$ is the number of variables of the lifted formula. 

More generally, we prove the following lifting theorem that requires a gadget $g$ to have sufficiently small correlation with all parities: we introduce a notion of hardness against ordinary decision trees that we call $(p,q)$-DT hardness. We show that if a formula $\Phi$ is $(p,q)$-DT hard, then every $\text{Res}(\oplus)$ refutation of size $s$ of the lifted formula $\Phi \circ g$ must be $\Omega\big(pq/\log(s)\big)$-deep. Our concrete lower bound is obtained by showing that Tseitin formulas on constant degree-expanders are $\big(\Omega(m),\Omega(m))$-hard.

An important ingredient in our work is to show that arbitrary distributions \emph{lifted} with such gadgets fool \emph{safe} affine spaces, an idea which originates in the earlier work of Bhattacharya, Chattopadhyay and Dvorak \cite{BCD24}.

\end{abstract}

\thispagestyle{empty}
\clearpage
\addtocounter{page}{-1}
\newpage

\tableofcontents
\thispagestyle{empty}
\clearpage
\addtocounter{page}{-1}

\section{Introduction}
\input{introduction}

\section{Preliminaries} \label{sec:preliminaries}
\input{preliminaries}
\section{A More Detailed Overview} \label{sec:intuition}

\input{high_level_overview}

\section{Organization of the Rest of our  Paper} \label{sec:overview}
\input{organization}

\section{Conditional fooling lemma}\label{sec:conditional_fooling}

\input{equidistribution_official}

\section{Description of CNF} \label{sec:CNF}
\input{CNF_description}

\section{The Utility of $(p,q)$-PDT Hardness} \label{sec:randomwalk}

\input{random_walk_with_restarts_official}

\section{Lifting DT-hardness to PDT-hardness} \label{sec:lifting_from_DT_hardness}
\input{defining_DT_hardness}
\input{lifting}

\section{Proving DT Hardness} \label{sec:proving_DT_hardness}

\input{proving_dt_hardness}

\section{Putting everything together}\label{sec:final_result}

\input{putting_everything_together}

\printbibliography

\begin{appendices}
\section{Facts about Amortized Closure} \label{app:amortized_closure}

\input{appendix_A}

\end{appendices}
\end{document}

%% file: introduction.tex
One of the simplest proof systems in propositional proof complexity is Resolution. Haken \cite{Haken85} obtained the first super-polynomial lower bounds on the size of proofs in this system for a CNF encoding of the pigeon-hole-principle forty years ago. Since then it has been very well studied with many beautiful results (see for example \cite{U87}, \cite{BSW01}, \cite{ABRW04}). Yet, seemingly slight strengthenings of resolution seem to frustrate current techniques in obtaining non-trivial lower bounds. We will consider one such strengthening, that was introduced by Itsykson and Sokolov \cite{IS14}, about ten years ago. This system is called resolution over parities, denoted by $\text{Res}(\oplus)$. It augments resolution by allowing the prover to make $\mathbb{F}_2$-linear inferences, while working with $\mathbb{F}_2$-linear clauses. Proving superpolynomial lower bounds for $\text{Res}(\oplus)$ remains a challenge. It is easy to see that $\text{Res}(\oplus)$ is a subsystem of $AC^0[2]$-Frege. While we know strong lower bounds for $AC^0$-Frege (see for example \cite{BIKPPW92}), obtaining super-polynomial lower bounds for $AC^0[2]$-Frege for any unsatisfiable formula in CNF would be a major breakthrough (see for example \cite{MP97}). Thus, $\text{Res}(\oplus)$ is in some sense the weakest natural subfragment of $AC^0[2]$-Frege for which proving strong lower bounds has remained challenging.  It is also natural to hope that a successful resolution of this challenge would give rise to new techniques and insight for taking on $\text{AC}^0[2]$-Frege. 
\newline

Itsykson and Sokolov proved exponential lower bounds on the size of tree-like $\text{Res}(\oplus)$ proofs using customized arguments for some formulas. General techniques for tree-like $\text{Res}(\oplus)$ proofs were developed in the independent works of Beame and Koroth \cite{BK23} and Chattopadhyay, Mande, Sanyal and Sherif \cite{CMSS23} that lifted lower bounds on the height of ordinary tree-like resolution proofs of a formula to that of the size of tree-like $\text{Res}(\oplus)$ proofs of the same formula lifted with an appropriate gadget. 
A more recent line of work (\cite{EGI24}, \cite{BCD24}, \cite{AI25}, \cite{EI25}) has focused on proving lower bounds against subsystems of $\text{Res}(\oplus)$ that are stronger than tree-like but weaker than general $\text{Res}(\oplus)$. Gryaznov, Pudlak and Talebanfard \cite{GPT22} had proposed several notions of regular proofs for the $\text{Res}(\oplus)$ system as appropriate first target for proving lower bounds. Efremenko, Garlik and Itsykson \cite{EGI24} established lower bounds against such a subsystem of $\text{Res}(\oplus)$ known as bottom-regular $\text{Res}(\oplus)$. Bhattacharya, Chattopadhyay and Dvorak \cite{BCD24} exhibited a CNF which is easy for resolution but hard for bottom-regular $\text{Res}(\oplus)$ - thereby strictly separating unrestricted $\text{Res}(\oplus)$ from bottom-regular $\text{Res}(\oplus)$. Subsequently, Alekseev and Itsykson \cite{AI25} significantly extended the reach of techniques by showing $\text{exp}(N^{\Omega(1)})$ lower bounds against $\text{Res}(\oplus)$ refutations whose depth is restricted to be at most $O(N \log \log N)$, where $N$ is the number of variables of the unsatisfiable CNF. This depth restriction was further improved to $O(N \log N)$ by Efremenko and Itsykson \cite{EI25}. \newline

A natural way towards proving lower bounds for unrestricted $\text{Res}(\oplus)$ would be improving the depth restriction all the way to $N^{\omega(1)}$. However, the techniques of Efremenko and Itsykson \cite{EI25} seem to get stuck at $O(N \log N)$. Efremenko and Itsykson \cite{EI25} posed the natural open problem of proving superpolynomial lower bounds against $\text{Res}(\oplus)$ refutations whose depth is restricted to $O(N^{1+\epsilon})$ where $\epsilon>0$ is some constant . \newline

Our main result, stated below, achieves such a bound.

\begin{theorem} \label{tseitin_lower_bound}
    Let $\Phi$ be the Tseitin contradiction on a $(m,d,\lambda)$ expander with $ \lambda < 0.95$ a small enough constant and $m$ odd. Let $n= md/2$ be the number of edges (which is also the number of variables in $\Phi$). Let $\text{IP}$ be the inner product gadget on $b= (4+\eta) \log(n)$ bits where $\eta > 0$ is an arbitrarily small constant. Let $\Psi = \Phi \circ \text{IP}$ be the lift of $\Phi$ by $\text{IP}$. Let $N= nb$ be the number of variables in $\Psi$. Then, any $\text{Res}(\oplus)$ refutation of $\Psi$ of depth $\leq O(N ^{2-\epsilon})$ requires size $\exp(\tilde{\Omega}(N^{\epsilon}))$
\end{theorem}

This pushes the frontier of depth of proofs against which super-polynomial lower bounds on size for $\text{Res}(\oplus)$ can be obtained, from $O(N \log(N))$ to $\tilde{O}(N^2)$. In particular, our lower bound achieves the best possible lower bound of $N^2$ on depth of efficient proofs that can be obtained using the random-walk-with-restart framework of Alekseev and Itsykson \cite{AI25}. Another way of interpreting our result is to say that any $\text{Res}(\oplus)$ proof of the hard formula $\Psi$ of size $\text{exp}(N^{o(1)})$ has to be almost $N^2$ deep, which is significantly super-critical.  \newline

\subsection{Brief Overview of Our Technique} \label{sec:brief}
Our work combines the approaches of Alekseev and Itsykson \cite{AI25}, Efremenko and Itsykson \cite{EI25} and Bhattacharya, Chattopadhyay and Dvorak \cite{BCD24} - along with a new equidistribution lemma for \textit{safe} affine spaces. \newline

Currently, the only known technique that can prove super-polynomial lower bounds on the size of DAG-like $\text{Res}(\oplus)$ proofs of even just linear depth is the very recent random-walk-with-restart method of Alekseev and Itsykson \cite{AI25}.
To understand the main innovation of our work, we first outline how a barrier of handling proofs deeper than $N\log N$ shows up naturally using this method. In fact,  this barrier shows up even when we try implementing the random-walk-with-restart method on ordinary resolution proof DAGs. Let us, therefore, first quickly review this technique using ordinary resolution DAGs.

The main idea of the random-walk-with-restart is as follows: consider a hard CNF formula $\Phi$ like the Tseitin contradiction on a $d$-regular expander graph. For such a $\Phi$ it is known that if we take a random assignment to its $N$ input variables, then even after making $\alpha N$ ordinary queries to the input variables, with probability  at least $p = 2^{-O(N/d)}$ a decision tree would not be able to locate a falsified clause of $\Phi$. Therefore, a walk of length $\alpha N$ on an ordinary resolution proof DAG starting from the root, upon given such a random assignment would  reach a node that is labeled by a sub-cube whose fixed variables do not reveal any falsified clause of $\Phi$, with probability $p$. We call such a sub-cube \emph{good}. However, if there as few as $s$ nodes in the whole DAG, then by union bound there must exist a good node $v$ which is reached by at least a $p/s$ fraction of all assignments. But then a simple counting argument shows that the co-dimension (i.e. the number of input variables fixed) of the cube associated with $v$, denoted by $A_v$, is at most $\log(s/p)$. If $d$, the degree of the expander $G$, is $O(\log n)$ and $m$ is the number of the vertices of $G$, and $s$ is at most $2^{m}$, then the co-dimension of $A_v$ is $O(N/\log N)$. That is although in the walk we may have queried $\alpha N$ edges in total, the DAG is constrained to forget many of them and finally  remembers only $O(N/\log N)$ many of them with sufficient  probability. The beautiful idea of Alekseev and Itsykson is that we could then make a fresh start of our walk from this node $v$, i.e. now sample a random assignment from $A_v$. They were able to show roughly the following: as long as the co-dimension of $A_v$ is smaller than $\beta N$ for some constant $\beta > 0$ and the cube $A_v$ is good, a random walk of length $\alpha N$ starting from $v$ reaches a good node with probability $p$ . Doing the whole co-dimension counting argument again, we conclude that there exists a good node $w$, such that $\text{co-dim}(A_w) \le \text{co-dim}(A_v) + \log(s/p)$. Hence, one could repeat this argument $\log N$ times as long as $s = 2^{o(N/\log N)}$. As each time we take $\alpha N$ steps, this ensures that there exists a node in the proof DAG at depth $\Omega(N\log N)$. With many more technical ideas from linear algebra like the notion of a \emph{closure} of a linear system on lifted variables, Alekseev and Itsykson were able to lift this idea to $\text{Res}(\oplus)$ DAGs where cubes were replaced by affine spaces and ordinary queries were replaced by parity queries. But there was a loss incurred in the process and they could only prove a depth lower bound of $\Omega(N\log \log N)$. This was improved to match the lower bound of $\Omega(N \log N)$, the best possible given the low success probability $p$ of reaching a good node, using more sophisticated notions like \emph{amortized closure} of a lifted linear system by Efremenko and Itsykson \cite{EI25}. 

We observe that if we sample a uniformly random input, then even starting from the root of an ordinary resolution DAG, the success probability of reaching a good node on making $\alpha N$ queries is no larger than $2^{-\Omega(N/d)}$. This is because the DAG could make $d$ queries on edges incident to a single vertex in $G$. The probability that the clause of $\Phi$ associated with that vertex is falsified is then precisely 1/2, even conditioned on previous queries. Thus, if the DAG processes $t$ vertices of $G$, expending $td$ queries, the probability of not falsifying any clause is at most $2^{-t}$. And unless we're able to boost the success probability $p$, we cannot handle depth beyond $N \log N$ even for ordinary resolution!

Armed with this observation, we design non-uniform distributions $\mu_v$, one for each good cube $A_v$ of small co-dimension, so that the success probability of reaching a good node $w$ starting from node $v$ is boosted all the way to a constant like $1/3$. The design of this distribution and the analysis of the associated random walk is non-trivial. But even after doing this, there are two challenges. First, the co-dimension counting argument goes for a toss! We were not able to find a simple way around this and fixing this is the most involved and non-trivial contribution of our work. The way we get around this is by going to a lifted space where the lifting is by not a stifling gadget as was being done in almost all previous work on $\text{Res}(\oplus)$, but with a gadget that has small correlation with all parities, like Inner-Product. Roughly speaking, we show that over such lifted spaces, any \emph{lifted distribution} restores the co-dimension counting argument in the sense that if we replace the co-dimension of an affine space by the amortized closure of an affine space, then its growth with re-starts proceeds more or less similar to how co-dimension's growth happened under the uniform distribution. This is driven by our main technical ingredient called the Conditional Fooling Lemma, stated and proved in Section~\ref{sec:conditional_fooling}. Finally, the second challenge is to lift the analysis of the random walk over ordinary resolution DAG under suitable family of non-uniform distribution to the analysis of a random walk over a $\text{Res}(\oplus)$ DAG under a lifted distribution. Interestingly, overcoming this second challenge is crucially aided by an equidistribution property (Lemma~\ref{exponential_sum}) that we establish for lifted affine spaces, en route to proving the Conditional Fooling Lemma.  This equidistribution property of lifted affine spaces is independently interesting.  Even more, it turns out that this lifting can be done in a very general way that doesn't depend on the specific base formula, in our case the Tseitin formula, but simply follows from the propery of the lifted space of the gadget.

A bit more concretely, we say that a formula $\Phi$ is $(p,q)$-DT hard if it satisfies the following: there exists a set of partial assignments (or equivalently cubes) $P \subseteq 
\{0,1,*\}^n$ which is downward closed, no $\rho \in P$ falsifies any clause of $\Phi$, and even more, for each $\rho \in P$ that fixes at most $p$ variables there exists a distribution $\mu_{\rho}$ over assignments consistent with $\rho$ that is hard for every ordinary decision tree $T$ of depth $q$ in the following sense: if we sample an input from $\mu_q$, the partial assignment obtained by additionally fixing the variables that $T$ queried,  also lies in $P$ with probability at least $1/2$. What we are able to show in Section~\ref{sec:lifting_from_DT_hardness} is that, the lifted formula $\Phi \circ g$ becomes `analogously' hard for parity decision trees, when $g$ has sufficiently small correlation with parities. We call this analogous hardness notion as $(p,q)$-PDT hardness. This notion was inspired by the analysis of certain combinatorial games in the work of Alekseev and Itskyson \cite{AI25}. In Section~\ref{sec:randomwalk}, we show that this notion of PDT-hardness yields lower bounds on depth of $\text{Res}(\oplus)$ proofs when its size is small. Chaining together the pieces yields the following lifting theorem, somehwat informally stated.

\begin{theorem}[Informal version of Theorem~\ref{ref: DT_hardness_implies_lower_bounds}]  \label{thm:lifting}
Let $\Phi$ be $(p,p+q)$-DT hard CNF. Then, amy $\text{Res}(\oplus)$ refutation of $\Phi \circ g$ of size $s$ must have depth $\Omega\left(\dfrac{pq}{\log s}\right)$, provided $g$ has sufficiently small size and small correlation wrt all parities.
\end{theorem}

It is worth stressing that the only place that the specifics of the CNF contradiction $\Phi$, which is the Tseitin formula over a constant-degree suitably expanding graph, enters into the argument is at proving that it's $\big((\Omega(m),\Omega(m)\big)$-DT-hard, which we do by a novel argument at the very end in Section~\ref{sec:proving_DT_hardness}.

\subsection{Some Other Related Work}  \label{Related-Work}
Our work makes use of the notion of amortized closure that was introduced by Efremenko and Itsykson \cite{EI25}. Apart from improving the depth lower bounds of small size $\text{Res}(\oplus)$ proofs, \cite{EI25} used this notion to give an alternative proof of a lifting theorem of Chattopadhyay and Dvor{\'a}k \cite{CD25} and their proof works for a broader class of gadgets. The lifting theorem is used in \cite{CD25} to prove super-criticial tradeoffs between depth and size of tree-like $\text{Res}(\oplus)$ proofs. 

Our work also crucially uses lifted distributions to boost the success probability of random walks with restarts. In particular, it uses an analytic property of the gadget to argue equidistribution of pre-images in a \emph{safe} affine space in the lifted world of a point $z \in \{0,1\}^n$ in the unlifted world. Such equidistribution, albeit wrt rectangles, have been earlier implicitly proved (see for example \cite{GLMWZ16,CFKMP21}) as well as explicitly proved in \cite{CDKLM17}. The analytic property of the gadget used in these works was essentially small discrepancy wrt rectangles (or being a 2-source extractor), something that seems to be significantly stronger than what we need of the gadget in this work.

\paragraph{Very recent work:} Soon after an initial version of our manuscript was uploaded, Byramji and Impagliazzo \cite{BI25a} proved a superlinear lower bound of $\Omega(N^{3/2-\epsilon})$ on the depth of $\text{Res}(\oplus)$ proofs of size $o\big(\text{exp}(N^{2\epsilon})\big)$ for the bit-Pigeon-hole Principle (BPHP). This work continued to use the random-walk-with-restart paradigm of \cite{AI25} under a nearly \emph{uniform distribution}, but there were two differences: the length of each walk was about $\sqrt{N}$ and the probability of the walk ending at a good node was as large as a constant. What prevented them from going beyond $\sqrt{N}$ length walk in each phase was a birthay paradox like phenomenon under their distribution.  Very recently, in a later version \cite{BI25b}, they seemed to have improved the depth bound for the BPHP to nearly quadratic, employing a more non-uniform distribution. The also prove a lifting theorem starting from our notion of $(p,q)$-DT-Hardness, using lifted distributions. However, the way they handle bottleneck counting  under lifted distributions seems very different from the way we do. In particular, they seem to be able to get away with simpler co-dimension fooling property of a lifted distribution, first observed in \cite{BCD24} in the context of a walk without re-start, even when re-starting random walks. This way of bottleneck counting apparently allows them to prove their lifting theorem, similar to our Theorem~\ref{thm:lifting}, but with just constant-size gadgets. This reduced size allows them to get a nearly quadratic lower bound on depth even in terms of formula size for a lifted Tseitin formula, whereas our lower bound is quadratic in just number of variables of the CNF. It is worth remarking that for this last result, the latest version of the Byramji-Impagliazzo work uses our result, proved in Section~\ref{sec:proving_DT_hardness}, that Tseitin formulas are $(\Omega(n),\Omega(n))$-DT-hard.

%% file: preliminaries.tex
\subsection{General Notation}
\begin{itemize}
    \item For a probability distribution $\mu$, when we sample a point $x$ according to $\mu$, we denote it by $x \leftarrow \mu$. 
    \item When an input is sampled according to some distribution $\mu$ conditioned on lying in some set $S$, we denote the resulting conditional distribution by $\mu \cap S$. Throughout this paper, whenever we encounter such expressions, it will be guaranteed (and easy to see) that $\text{supp}(\mu) \cap S \neq \phi$, so the conditioning is valid.
    \item When $x$ is sampled according to uniform distribution over a set $T$, we denote it by $x \sim T$.
    \item Throughout this paper, we shall identify a linear form $\ell \in (\F^{nb})^{*}$ by its canonical representation wrt the standard inner product, i.e., $\ell(x)= \langle \ell, x \rangle$ where $\langle x , y \rangle= \displaystyle \sum_j x_j y_j \pmod{2}$.
    \item For a partial assignment $\rho \in \{0,1,*\}^n$ we denote
    $$ \text{free}(\rho) = \{ i \in [n]| \rho(i) = * \}$$
    $$ \text{fix}(\rho) = \{i \in [n]| \rho(i) \neq *\}$$
\end{itemize}

\subsection{Resolution over parities}
\begin{definition}
A linear clause $\ell_C$ is an expression of the form 
$$\ell_C(x)=[\langle \ell_1, x \rangle = b_1] \lor [\langle \ell_2, x \rangle = b_2] \cdots \lor [\langle l_k, x \rangle = b_k]$$
Here $x, \ell_1, \cdots, \ell_k \in \F^n$. Note that the negation of $\ell_C$, $\neg \ell_C$ is an affine space:
$$ \neg \ell_C= \{x \in \F^n | \langle \ell_1, x\rangle = 1-b_1, \cdots , \langle \ell_k, x \rangle= 1-b_k\}$$
\end{definition} Also notice that every ordinary clause is also a linear clause. \newline

$\text{Res}(\oplus)$ (defined in \cite{IS14}) is a proof system where every proof line is a linear clause. The derivation rules are as follows:
\begin{enumerate}
    \item \textbf{Weakening: }From $\ell_C$, derive $\ell_D$ where $\ell_D$ is any linear clause semantically implied by $\ell_C$ ($\ell_C \implies \ell_D$). \newline
    
    This step can be verified efficiently because it is equivalent to $\neg \ell_D \subseteq \neg \ell_C$, and this is a statement about containment of one affine space within another - which can be checked by Gaussian elimination.
    \item \textbf{Resolution: }From $\ell_C^{(1)}(x) = \ell_C(x) \lor [\langle \ell, x \rangle=b]$ and $\ell_C^{(2)}(x) = \ell_c(x) \lor [\langle \ell, x \rangle=1-b]$, derive $\ell_C(x)$
\end{enumerate}

A $\text{Res}(\oplus)$ refutation of a CNF $\Phi$ starts with the axioms being the clauses of $\Phi$ (which, as noted above, are also linear clauses) and applies a sequence of derivation rules to obtain the empty linear clause $\emptyset$. \newline

\subsubsection*{Affine DAGs}

For an unsatisfiable CNF $\Phi$ define the search problem 
$$\text{Search}(\Phi) = \{(x,C)| C \text{ is a clause of }\Phi, C(x)=0\}.$$ Just as a resolution refutation of $\Phi$ can be viewed as a cube-DAG for solving $\text{Search}(\Phi)$, a $\text{Res}(\oplus)$ refutation can be viewed as an affine-DAG for solving $\text{Search}(\Phi)$.

\begin{definition}
    An affine DAG for $\text{Search}(\Phi)$ is a DAG where there is a distinguished root $r$, each node $v$ has an associated affine space $A_v$, and each node has outdegree either 2, 1, or 0. Each outdegree 0 node $w$ is labelled with a clause of $\Phi$, $C_w$. The following requirements are satisfied:
    \begin{enumerate}
        \item If $v$ has two children $v_1, v_2$ then $A_v= A_{v_1} \cup A_{v_2}$.
        \item If $v$ has only one child $w$, then $A_v \subseteq A_w$.
        \item If $v$ has no children, then for any $x \in A_v$, $C_v(x)=0$ where $C_v$ is the clause labelled on $v$.
        \item The affine space labelled on the root is the entire space $\F^n$.
    \end{enumerate}
\end{definition}
A $\text{Res}(\oplus)$ refutation for $\Phi$ can be viewed as an affine DAG for $\text{Search}(\Phi)$ by viewing the sequence of derivations as a DAG: for each node, the associated affine space is the negation of the linear clause derived at that node. The leaves are the axioms - the clause labelled at each leaf is simply the corresponding axiom. \newline

We classify nodes based on their outdegree as follows.

\begin{enumerate}
    \item A node with no children is called a leaf.
    \item A node with one child is called a \textit{weakening node}. (Because in the $\text{Res}(\oplus)$ refutation this node was derived by weakening.)
    \item Let $v$ be a node with two children $v_1, v_2$. In this case it holds that $A_v= A_{v_1} \cup A_{v_2}$; $A_{v_1} = A_v \land [\langle \ell, x \rangle=b]$ and $A_{v_2} = A_v \land [\langle \ell, x \rangle=1-b]$ for some $\ell \in \F^n$. Such a node is called a \textit{query node}; we say the affine DAG queries $\ell$ at node $v$. (In the $\text{Res}(\oplus)$ refutation, node $v$ was obtained by resolving the linear form $\langle \ell, x \rangle$.)
\end{enumerate}

\subsubsection*{Path of an input}
Here we consider any affine DAG that arises from some $\text{Res}(\oplus)$ refutation. For any node $v$ and any $x \in A_v$, we define the path of $x$ starting from $v$ as follows:
\begin{itemize}
    \item Start with the currrent node $v$.
    \item If the current node is $w$ and $w$ has no children, terminate the path.
    \item If the current node is $w$ and has two children $w_1, w_2$, we know that $A_w= A_{w_1} \cup A_{w_2}$. In this case it will hold that $A_{w_1} = A_w \cap \{\tilde{x} | \langle \ell, \tilde{x} \rangle=b\}$ and $A_{w_2} = A_w \cap \{\tilde{x}| \langle \ell, \tilde{x} \rangle = 1-b\}$ for some $\ell \in \F$. If $\langle \ell, x\rangle=b $, the next node in the path is $w_1$. Otherwise, the next node in the path is $w_2$
    \item If the current node $w$ has only one child $w_1$, the next node in the path is $w_1$. 
\end{itemize}

The way the path is defined ensures that if the path of $x$ visits the node $w$, $x \in A_w$. Consequently, for any $x$ and $v$ such that $x \in A_v$, the path of $x$ starting from $v$ visits a leaf whose clause is falsified by $x$. \newline

In particular, for any $x$, if we follow the path traversed by $x$ from the root, we end up at a clause falsified by $x$. \newline

\begin{definition}
We define the length of a path to be the number of query nodes encountered on the path. (The weakening nodes do not contribute to the length.)
\end{definition}

\begin{definition}
    The depth of a node $v$ is the largest length of a path from the root to $v$. The depth of the refutation is the depth of the deepest node.
\end{definition}

\subsection{Notations about lifted spaces}
In this paper, we shall be working with a gadget $g: \F^b \rightarrow \F$. The base space will be $\F^n$. The lifted space will be $\F^{N}$ where $N=nb$. The coordinates of the lifted space are $\{(i,j) | i \in [n], j \in [b]\}$.

 \begin{definition} \label{def: block of i}The set of coordinates $\{x_{i,j} | j \in [b]\}$ is called to be \textit{the block of $i$}. The $i$-th block will be denoted as $x(i) \in \F^b$ \end{definition}

 The gadget $g$ naturally induces a function $g^n: \F^{nb} \rightarrow \F^n$ by independent applications of $g$ on the $n$ different blocks. We shall abbreviate $g^n$ by $G$.

    \begin{definition}
    \label{def: G(beta)}
    For any assignment $\beta \in \F^{S \times [b]}$ to the variables in blocks of $S$, we define the partial assignment $G(\beta) \in \{0,1,*\}^n$ as follows:
    $$
    G(\beta)_i= \begin{cases}
        * & \text{ if } i \not \in S \\
        g(\beta(i)) & \text{ otherwise}
    \end{cases}
    $$
    \end{definition}

    \begin{definition} \label{def: G inverse}
    For any assignment $z \in \{0,1\}^n$ define $G^{-1}(z)$ to be the set of preimages of $z$:
    \begin{itemize}
        \item Then, $G^{-1}(z) \subseteq \F^{nb}$: 
        $$ G^{-1}(\alpha) = \{y | y \in \F^{nb}, g(y(i))= z_i \text{ }\forall i \in [n]\}$$
    \end{itemize}
    \end{definition}
\begin{definition} \label{def: lifted distribution}
    For any distribution $\mu$ on $\F^n$ define the lifted distribution $G^{-1}(\mu)$ on $\F^{nb}$ as the outcome of the following sampling procedure:
    \begin{enumerate}
        \item Sample $z \leftarrow \mu$.
        \item Sample $x$ uniformly at random from $G^{-1}(z)$.
    \end{enumerate}
    Any distribution of the form $G^{-1}(\mu)$ is called a \textit{lifted distribution}.
    
    \end{definition}
    \begin{definition}
        For a partial assignment $y \in \F^{S \times [b]}$ to some set of blocks in the lifted space, we define $G(y) \in \F^S$ to be the corresponding partial assignment in the unlifted world: $G(y)_i= g(y(i))) \text{ } \forall i \in S $.
    \end{definition}
    \textbf{Caution: }$G$ will also be used to denote a graph in one part of the paper. It should not cause confusion, however, because in that section we will not be using this interpretation of $G$.

    \begin{definition}
        For a partial assignment $y \in \F^{S}$ ($S \subseteq N$), we define the \textit{cube of }$y$, $C_y$ to be the set of points consistent with $y$:

        $$C_y = \{ x \in \F^N | x_i=y_i \text{ } \forall i \in S \}$$
    \end{definition}

    \begin{definition} \label{def: extendable}
    For an affine space $A \subseteq \F^{nb}$ and a partial assignment $y \in \{0,1,*\}^{nb}$, call $y$ $A$-extendable if there exists $x \in A$ consistent with $y$
    \end{definition}

    \begin{definition} \label{def A_y}
        For an affine space $A \subseteq \F^{nb}$ and an extendable partial assignment $y \in \F^{S}$ (where $S \subseteq [nb]$) define $A_y \subseteq \F^{[nb]\setminus S}$ as follows:
        $$A_y = \{\tilde{x} | (\tilde{x},y) \in A\}$$
    \end{definition}

\subsection{Linear algebraic facts about lifted spaces}


In this subsection we import facts about closure and amortized closure proved by \cite{EGI24}, \cite{AI25} and \cite{EI25}.

\begin{itemize}
    \item \textbf{Safe set of linear forms} \begin{definition} \textbf{(from \cite{EGI24})} \label{def: safe}
    A set of linear forms $V= \{\ell_1, \ell_2, \cdots , \ell_m \} \subseteq \F^{nb}$ \footnote{Technically, a linear form $\ell $ should lie in the dual space $(\F^{N})^{*}$. In this, we identify a linear form $\ell$ as an element of $\F^{N}$ by its canonical representation w.r.t the standard inner product: $\ell(x)= \langle \ell, x \rangle$}. is \textit{safe} if for any $k$ linearly independent forms $w_1, w_2, \cdots , w_k \in \text{span}(S)$, $\text{supp}(w_1) \cup \text{supp}(w_2) \cup \cdots \cup \text{supp}(w_k)$ includes at least $k$ distinct blocks.
     
    \end{definition}

    \item  \textbf{Equivalent definition of \textit{safe}:} Let 
    $$ M= \begin{bmatrix}
        \ell_1 \\ \hline  \ell_2 \\ \hline \cdots \\ \hline \ell_m 
    \end{bmatrix} \in \F^{m \times nb}$$
    Let $r= \text{rank}(M)$. $V$ is nice iff there exist indices $c_1, c_2, \cdots , c_r \in [nb]$, each lying in different blocks, such that the set $\{Me_{c_1}, \cdots , Me_{c_r}\} \subseteq \F^{m}$ is linearly independent. ($Me_j$ is the $j$-th column of $M$)
    The proof of equivalence of these two definitions can be found in Theorem 3.1 in \cite{EGI24}. 
    
     \begin{fact} \label{fact: depends only on span}
        Whether or not a set of linear forms is safe depends only on their span. This is clear from the second equivalent definition.
    \end{fact}

    \item \textbf{Safe affine spaces: } \begin{definition} \label{safe affine space}
        Let $A \subseteq \F^{nb}$ be an affine space. Let $A= \{x | Mx=b\}$. Then, $A$ is called a \textit{safe affine space} if and only if the rows of $M$ are safe (i.e., the set of linear forms defining $A$ is safe). This does not depend on a specific choice of $M$ by Lemma 4.1 in \cite{EGI24}. 
        
    \end{definition}

    \item \textbf{Deviolator: }
    For a subset of the blocks $S \subseteq [n]$ and a linear form $\ell \in \F^{nb}$, define $ \ell[\setminus S] \in \F^{(n-|S|)b} $ to be the projection of $v$ on the coordinates of $[n] \setminus S $. \\
    \begin{definition}
    A subset $S \subseteq [n]$ is a deviolator for $V= \{\ell_1, \ell_2, \cdots , \ell_m\} \subseteq \F^{nb}$ if $\{\ell_1[\setminus S], \ell_2[\setminus S], \cdots , \ell_m[\setminus S] \}$ is a nice set of linear forms.
    \end{definition}

    \item \textbf{Closure of a set of linear forms:}
    \begin{definition} \textbf{(from \cite{EGI24})}
        Closure of a set of linear forms $V= \{\ell_1, \ell_2, \cdots , \ell_m\}$ is the minimal deviolator for $V$. (It is known that this deviolator is unique, and also it depends only on $\text{span}(V)$ - Lemma 4.1 in \cite{EGI24}.)
    \end{definition}

    \item \textbf{Closure of an affine space: }

    \begin{definition}
    For an affine space $A$ given by the set of equations $A= \{x| Mx=b\}$, define $\Cl(A)$ to be the closure of the set of rows of $M$ (i.e., $\Cl(A)$ is the closure of the set of defining linear forms of $A$). This does not depend on a specific choice of $M$ by Lemma 2.11 in \cite{EI25}.

    \end{definition}
   
    \item \textbf{Closure Assignment}
    \begin{definition} \label{closure_assignment}
        For an affine space $A$, a \textit{closure assignment} $y$ is any assignment to the the coordinates in $\Cl(A) \times [b]$: $y \in \F^{\Cl(A) \times [b]}$.
    \end{definition}

    \item \textbf{Amortized Closure of a set of linear forms}
    \begin{definition} \textbf{(from \cite{AI25})}
 Let $V= \{\ell_1, \ell_2, \cdots , \ell_k\} \in \F^{nb}$ be a set of linear forms. We define $\Clh(V) \subseteq [n]$ as follows: Let 
          $$ M= \begin{bmatrix}
            v_1 \\ 
            \hline \\
            v_2 \\ 
            \hline \\
            \cdots \\
            \hline \\
            v_t
        \end{bmatrix}$$
    Call a set of blocks $S = \{s_1, s_2, \cdots,  s_k\} \subseteq [n]$ \textit{acceptable} if there exist columns $c_1, c_2, \cdots , c_{k}$, such that $c_j$ lies in block $s_j$ and the set $\{ Me_{c_1} , Me_{c_2}, \cdots , Me_{c_{k}}\}$ is linearly independent. The amortized closure of $V$, $\Clh(V)$, is the lexicographically largest acceptable set of blocks. \newline
    
    It is known that $\Clh(V)$ depends only on $\text{span}(V)$ (Lemma 2.11 in \cite{EI25})
        
    \end{definition}

    \item \textbf{Amortized Closure of An Affine Space}
    \begin{definition}
        Let $A \subseteq \F^{nb}$ be an affine space; $A= \{x | Mx=b\}$. The amortized closure of $A$, $\Clh(A)$, is defined to be the amortized closure of the set of rows of $M$. This does not depend on a specific choice of $M$ (Lemma 2.11 in \cite{EI25})
    \end{definition}

\end{itemize}

Now we import some facts and lemmas about closure and amortized closure from \cite{EGI24}, \cite{EI25} and \cite{AI25}.

    \begin{lemma}
        \label{lem: removing closure makes affine space nice}
     If $y$ is an $A$-extendable closure assignment, $A_y$ is a safe affine subspace. (Follows from definition.)
    \end{lemma}
\begin{lemma}For any affine space, $\Cl(A) \subseteq \Clh(A)$ (Lemma 2.15 in \cite{EI25})\end{lemma}
     \begin{lemma} \label{closure containment}
     If $A , B$ are affine spaces with $B \subseteq A,$ then $\Clh(A) \subseteq \Clh(B)$ (Corollary 2.19 in \cite{EI25}) and $\Cl(A) \subseteq \Cl(B)$ (Lemma 4.2 in \cite{EGI24})
    \end{lemma}
    \begin{lemma} \label{closure remains same}Let $V \subseteq \F^N$ be a set of linear forms with $\Cl(V)=S$. Let $W= V \cup \{e_{j,k}|j \in S , k \in [b]\} $. Then, $\Clh(V)= \Clh(W), \Cl(V)=\Cl(W)$. \end{lemma}
Lemma \ref{closure remains same} becomes clear once one examines the proof of Lemma 2.15 in \cite{EI25} closely. For completeness we include a self-contained proof in Appendix \ref{app:amortized_closure}.

     \begin{lemma} \label{lem: amortized_closure_continuity}
     Let $V \subseteq W \subseteq \F^{nb}$ be sets of linear forms with $|W|= |V|+1$. Then, $|\Clh(W)|\leq |\Clh(V)|+1$, and moreover, if $|\Clh(W)| = |\Clh(V)|+1$ then $\Cl(W)=\Cl(V)$ (Theorem 2.18 and Lemma 2.17 in \cite{EI25}).
    
    \end{lemma}
 
We now state a useful corollary of the above. \newline
\begin{corollary} \label{both_are_nice} Let $B \subseteq A$ be affine spaces such that $\text{codim}(B)= \text{codim}(A)+1$ and $|\Clh(B)| = |\Clh(A)|+1$. Let $y$ be any $A$-extendable closure assignment. Then, $A_y, B_y$ are both nice affine subspaces and $\text{codim}(B_y)= \text{codim}(A_y)+1$.
\end{corollary}
A proof of Corollary \ref{both_are_nice} is included in Appendix \ref{app:amortized_closure}.

%% file: high_level_overview.tex
At a high level, our proof combines the approaches of Alekseev and Itsykson \cite{AI25}, Efremenko and Itsykson \cite{EI25} and Bhattacharya, Chattopadhyay and Dvorak \cite{BCD24}. It does so by boosting the success probability of the `random walk with restart' method of \cite{AI25} by sampling inputs from a lifted distribution. The idea of using lifted distribution to do random walks appeared in \cite{BCD24}. The bottleneck counting uses the notion of amortized closure instead of codimension of an affine space as done in \cite{EI25}. However, combining these approaches requires significant new ideas -- along with a new equidistribution lemma for gadgets with sufficiently small Fourier coefficients (Lemma \ref{exponential_sum}).
In this section we give a brief overview of how these approaches fit together.

\paragraph{Approach of Alekseev and Itsykson \cite{AI25}} Recall the overall idea of the random-walk-with-restart as outlined for ordinary Resolution proof DAGs in Section~\ref{sec:brief}. We start by describing how this approach was implemented by \cite{AI25} for $\text{Res}(\oplus)$ DAGs by developing required linear algebraic machinery, and why it fell short of handling DAGs of depth $O(N \log N)$ that was relative easily possible for resolution DAGs.
The main idea in \cite{AI25} is this: they take the CNF $\Psi$ to be Tseitin contradiction over an $(n, \log(n), O(\log(n))$-expander lifted with an appropriate gadget; they assume we are given a size $s$ Res-($\oplus$) refutation $\Pi$ of $\Psi$, and they locate a path of length $n \log \log (n)$ in $\Pi$. They do this inductively: at Phase $j$, they locate a vertex $v_j$ at depth $\Omega(nj)$. Given this vertex $v_j$, they show that as long as $\text{codim}(A_{v_j})$ is not too large, there is another vertex $v_{j+1}$ which is at distance $\Omega(n)$ from $j$. They show they can inductively find one more vertex as long as $j \leq O(\log \log(n))$ - and this gives the depth lower bound. \newline

Let us describe it in a bit more detail. Alekseev and Itsykson carefully choose a set of partial assignments in the unlifted world, $P \subseteq \{0,1,*\}^n$ with the idea that any partial assignment $\rho \in P$ leaves \emph{some uncertainty}  about which clause  of the unlifted Tseitin formula would be falsified if one were to extend $\rho$ at random to a full assignment. \newline 

In the phase $j$, Alekseev and Itsykson \cite{AI25} have located a vertex $v_j$ at depth $\Omega(jn)$. They want the codimension of $A_{v_j}$, the affine space that the proof $\Pi$ associates with $v_j$, to be small ( $\leq O(j (\log(s/p)) (b+1)^j)$, which is less than $jn$ when $j$ is small enough; $p$ is a parameter we shall specify soon). Small co-dimension implies a small closure, i.e. $\text{codim}(A_{v_j}) \ge |\Cl(A_{v_j})|$. We assume that variables in the unlifted world that correspond to blocks in $\Cl(A_{v_j})$ are revealed, but variables that correspond to blocks outside of the closure, i.e. in $[n]-\Cl(A_{v_j})$ are yet not revealed. Hence, Alekseev and Itsykson fix an extendable closure assignment $y_j \in \F^{\Cl(A_{v_j}) \times [b]}$ such that $G(y_j)$ lies in $P$. 
They show (using a combinatorial argument) that the following happens  when we uniformly sample a point $x \in A_{v_j} \cap C_y$ and follow the path of $x$ from $v_j$ for $\Theta(n)$ steps or until it lands at a leaf node, whichever happens earlier: let $w$ be the vertex reached. Let $\tilde{x} \in \F^{\Cl(A_w) \times [b]}$ be the restriction of $x$ to the variables of $\text{Cl}(A_w)$. Let $\rho \in \{0,1,*\}^n$ be the partial assignment that leaves all variables outside of $\Cl(A_w)$ free and $\rho|_{\Cl(A_w)} = G(\tilde{x})$. With probability at least $p$, $\rho$ is in $P$, i.e. this $\rho$ reveals little about where a potential falsified clause may be. Let us call such a node $w$ to be \emph{good}. For this combinatorial argument to work, it is essential that the starting partial assignment, $G(y_j)$ lies in $P$ and it does not fix too many bits: $|\Cl(A_{v_j})| \leq O(n/\log(n))$.  \newline

One such good $w$ will be the next node, $v_{j+1}$ - and the next closure assignment $y_{j+1}$ could be anything in $\mathbb{F}_2^{ \Cl(A_w) \times [b]}$ such that $G(y_{j+1}) \in P$ and $y_{j+1}$ is extendible in $A_w$. The existence of such a $y_{j+1}$ trivially follows as $x|_{\Cl(A_w) \times [b]}$ satisfies those requirements.
As $w$ is good, it cannot be a leaf node for no falsified clause can be identified at $w$. Hence, all good $w$'s are at distance $\Omega(n)$ from $v_j$ - so the only condition Alekseev and Itsykson need to maintain is that the codimension of $A_w$ is not too high. They show the existence of such a $w$ using a simple bottleneck argument: there exists a $w$ such that a uniformly random $x \in A_{v_j} \cap C_y$ reaches $A_w$ with probability $\geq p/s$ as there are at most $s$ many nodes at any given distance from node $v_j$. In particular, $|A_w| \geq \dfrac{p}{s} |A_{v_j} \cap C_y|$, which implies $\text{codim}(A_w) \leq (b+1) \text{codim}\big(A_{v_j}\big) + \log(s/p) \leq O((j+1) (\log(s/p)) (b+1)^{j+1})$. \newline

Let us now briefly explain why this approach fails to go beyond depth $O(n \log \log n)$. Once $\text{codim}(A_{v_j})$ exceeds $n/\log(n)$, the underlying combinatorial argument in \cite{AI25} to get the next node fails. Hence, the depth lower bound obtained by this argument depends on the number of iterations until which $\text{codim}(A_{v_j})$ is guaranteed to be less than $n/\log(n)$. In this case, there are two factors causing rapid growth of (the guaranteed upper bound on) $\text{codim}(A_{v_j})$: first, at each step, the codimension of the next node can increase geometrically. Second, the success probability $p$ in \cite{AI25} is pretty low: around $2^{-O(n/\log(n))}$ - this also contributes to the growth of the valid upper bound on $\text{codim}(A_{v_j})$. Please recall from Section~\ref{sec:brief}, that the estimate of the success probability of the random walk being $2^{-O(n/\log n)}$ is tight, even on an ordinary resolution DAG. 

\paragraph{Improvement to depth $\Omega(N\log N)$:} In 2025, Efremenko and Itsykson \cite{EI25} bypassed the first barrier (of the codimension growing geometrically at each step) by introducing a new notion, different from co-dimension, to track progress: this notion is the amortized closure $\Clh(A)$. Notice that the reason why the codimension could be growing geometrically in \cite{AI25} was that fixing the bits of $\text{Cl}(A_v)$ to $y$ adds $b |\text{Cl}(A_v)|$ more constraints, which could be as large as $b \times \text{codim}(A_v)$. One of the key lemmas in \cite{EI25} is that if $|\Clh(A_w)| = |\Clh(v)|+k$, then $\text{Pr}_{x \sim A_v \cap C_y} [x \in A_w] \leq 2^{-k} $. In other words, if $|\Clh(A_w)| = |\Clh(A_v)|+k$, among the equations defining $A_w$, there exist $k$ linearly independent equations and moreover, these equations are also linearly independent from the equations of $A_v \cap C_y$ as the properties of amortized closure ensure $\Clh(A_v) = \Clh(A_v \cap C_y)$. Now, \cite{EI25} runs the same argument again. This time, it yields the following recursion: $|\Clh(A_{v_{j+1}})| \le |\Clh(A_{v_j})| + \log(\frac{s}{p})$, which prevents a geometric growth on the size of the amortized closure (as was happening with codimension earlier).  This ensures that $|\Clh(A_{v_j})| \leq O(j \log(s/p))$ at Phase $j$. This enabled Efremenko and Itsykson \cite{EI25} to find a vertex at depth $\Omega(N \log(N))$ assuming $s$ was $\text{exp}(N^{o(1)})$. Thus, the same simple bound on depth achieved on an ordinary resolution DAG by walking under the uniform distribution was achieved for $\text{Res}(\oplus)$ DAGs albiet only after lifting the base Tseitin formula with an appropriately stifling gadget and using the sophisticated notion of an amortized closure of a lifted linear system. 
However, the second barrier still remained, as we argued earlier that it remained even for ordinary resolution: the success probability $p$ of each phase of the random walk was very small; around $2^{-n/\log(n)}$. Thus, this argument could not go beyond depth $N \log(N)$. 

\paragraph{Our approach for depth $N^{2-\epsilon}$:} One of the main contributions of this work is getting around this low success probability barrier. As explained in Section~\ref{sec:brief}, we do that by going to non-uniform distributions. However, as said before, the bottleneck counting argument goes for a toss (even for sub-cube DAGs of ordinary resolution). To restore the ability of bottleneck counting, a basic requirement seems that our non-uniform distribution should at least fool the simplest linear algebraic notion of co-dimension. Such a fooling was devised in the work of Bhattacharya, Chattopadhyay, and Dvorak \cite{BCD24}. In \cite{BCD24}, the authors prove a separation between a restricted class of $\text{Res}(\oplus)$ refutations (known as bottom-regular refutations) and general $\text{Res}(\oplus)$ refutations. Their proof also employed a bottleneck argument, but instead of sampling from the uniform distribution, they were sampling from a lifted distribution. The key observation in \cite{BCD24} was that if $g: \F^b \rightarrow \{0,1\}$ is an appropriate gadget, then for any lifted distribution $\bar{\mu}$ and any affine space $A$, $\text{Pr}_{x \leftarrow \bar{\mu}} [x \in A] \leq 2^{-\Omega(\text{codim}(A)/b)}$. This would be sufficient, as it essentially was for \cite{BCD24}, to do bottleneck counting at the end of the first phase of the random walk, starting from the root of the DAG. \newline

To do bottleneck counting at subsequent restarts, one needs a conditional version of the above fooling statement in which one samples from the lifted fooling distribution, conditioned on the sampled input lying in some affine space $A$. Here, $A$ corresponds roughly to the affine space associated with a node of the DAG walk from where we need to re-start.
One might hope that conditional version is true: if $B \subseteq A$ are two affine spaces and $\tilde{\mu}$ is a lifted distribution, then $\text{Pr}_{x \leftarrow \tilde{\mu} } [x \in B | x \in A] \leq 2^{-\Omega(\text{codim}(B)-\text{codim}(A))/b}$. If this were true, we could modify the proof of \cite{AI25}: instead of sampling the input uniformly from $A_v \cap C_y$, we could sample from a lifted distribution tailored to our needs - which can hopefully boost the success probability. Unfortunately, such a statement cannot be true for any gadget, as the following counterexample shows. \newline

\textbf{Counterexample to a naive idea of conditional fooling} \newline
\begin{center}
\fbox{
\begin{minipage}{30em}
Let $t \in \F^b$ be a point, such that the first bit (wlog) is $g$-sensitive at $t$, i.e. $g(t) \neq g(t\oplus e_{\{1\}})$. WLOG, let $g(t)=0$. The equations for $A$ are as follows: for all $i \in [n], j \in [t] \setminus \{1\}, x_{ij}=t_j$. In $B$, we add the following extra equations: for all $i \in [n],$ $x_{i1}=t_1$. Let $\bar{\mu}$ be the uniform distribution on $G^{-1}(0^n)$. Then, even though $\text{codim}(B)= \text{codim}(A)+n$,
$$ \text{Pr}_{x \leftarrow \bar{\mu} } [x \in B | x \in A]=1$$
\end{minipage}
}
\end{center}

Intuitively, the reason why conditional fooling does not happen in this counterexample is that $A$ fixes too many linear forms in a block - and thus, when sampling from $G^{-1}(0^n) \cap A$, the distribution on each block is not controllable. One might imagine if the equations defining $A$ do not concentrate too much on any single block, the distribution $G^{-1}(z) \cap A$ behaves more nicely. One notion of the equations defining $A$ not concentrating on any single block is that $A$ is a safe affine space. Indeed, it turns out that the conditional fooling conjecture is actually true when $A$ and $B$ are both safe affine spaces (Lemma \ref{increase_by_one}) and the gadget $g$ is \emph{nice}, i.e. all Fourier coefficients of $g$ are sufficiently small in $n$, the size of the unlifted space.
 Given this, it is not hard to show that lifted distributions fool amortized closure. 
In particular, we shall show that if $\bar{\mu}$ is any $g$-lifted distribution, $g$ being nice, and $|\Clh(A_w)|= |\Clh(A_v)|+k$, then $\text{Pr}_{x \leftarrow \bar{\mu} \cap C_y} [x \in A_w | x \in A_v] \leq (3/4)^k$ (Conditional Fooling Lemma, i.e. Lemma \ref{conditional_fooling}). \newline

Thus, the Conditional Fooling Lemma crucially allows us to do bottleneck counting wrt the amortized closure of affine spaces at the end of a random walk after re-start from a node $v$, when the walk is being triggered by a suitably lifted distribution. However, we still need to find a concrete lifted distribution that would boost the success probability of the walk on a $\text{Res}(\oplus)$ DAG. Designing such a lifted distribution and analyzing the corresponding random walk is our second contribution. To do this, we first formulate a certain notion of hardness against \emph{parity decision trees} (PDT) of depth $q$ that are promised that the input comes from a $g$-lifted distribution, conditioned on lying inside an affine space $A$ (corresponding to the affine space of the node of the DAG from where a walk would restart) with $|\Clh(A)| \le p$ . We call this notion as $(p,q)$-PDT hardness, wrt the gadget $g$. This notion was inspired by the analysis of certain combinatorial games by Alekseev and Itsykson \cite{AI25}.  Section~\ref{sec:randomwalk} shows that $\text{Res}(\oplus)$ proofs of size $s$ for a lifted CNF $\Phi \circ g$  must have depth $\Omega\left(\frac{pq}{\log s}\right)$, whenever  $g$ is nice and $(\Phi,g)$ is $(p,q)$-PDT-Hard. This requires the use of the Conditional Fooling Lemma.  With its help, we have reduced the problem of proving size-depth tradeoff lower bounds for $\text{Res}(\oplus)$ DAGs to that of proving certain kind of hardness against PDTs wrt lifted distributions. \newline

At this point, it is natural to wonder if the problem can be reduced even further to proving appropriate hardness against ordinary decision trees (DT). This is exactly what we do after formulating such a notion of hardness. This notion we call $(p,q)$-DT-Hardness which was stated in Section~\ref{sec:brief} and is formally stated in Section~\ref{sec:lifting_from_DT_hardness}. Finally, we prove in Section~\ref{sec:lifting_from_DT_hardness} that if any CNF formula $\Phi$ is $(p,p+q)$-DT-hard, then the pair $(\Phi,g)$ is $(p,q)$-PDT Hard. For this reduction, we make use of our  Equidistribution Lemma, Lemma~\ref{exponential_sum}, which is independently interesting. \newline

In this way, we get a novel lifting theorem: $(p,q)$-DT hardness of any CNF $\Phi$ lifts to yield that size $s$ $\text{Res}(\oplus)$ proofs of $\Phi \circ g$ must have depth $\Omega\left(\frac{pq}{\log s}\right)$ as long as $g$ is a nice gadget whose block size is not too large. All that remains is to find a CNF that is appropriately DT hard and to find a $g$ that is nice and does not have a very large block size. The latter is simple to find: any bent function like Inner-product on $O(\log n)$ many bits suffices. To find the former, we show in Section~\ref{sec:proving_DT_hardness}, that a Tseitin formula defined over a $(n,d,\lambda)$ expander graph, for some constant $\lambda < 1$, is $(\Omega(n),\Omega(n))$-DT-Hard. This, finally, gives us the required formula for proving our main result of exponential lower bounds on the size of $\text{Res}(\oplus)$ proofs that have depth $N^{2-\epsilon}$ for any constant $\epsilon > 0$, where $N$ is the number of variables. The number of clauses of the lifted formula is $N^{O(1)}$. It is worth noting that this last result of ours has been directly used very recently by Byramji and Impagliazzo \cite{BI25b}.

%% file: organization.tex
We first describe the CNF in our final result in Section \ref{sec:CNF}. We prove our final result (Theorem \ref{tseitin_lower_bound}) in Section \ref{sec:final_result}, which uses machinery developed in Sections \ref{sec:conditional_fooling}, \ref{sec:randomwalk}, \ref{sec:lifting_from_DT_hardness}, and \ref{sec:proving_DT_hardness}. The following figure, denoting the dependency graph of the various components needed for our main result, depicts the organization of our work.

\input{tikz_dag}

 \vspace{2mm}
Section \ref{sec:proving_DT_hardness} can be read independently of all other sections.

%% file: tikz_dag.tex
\begin{tikzpicture}[scale=.2, thick, auto=left,every node/.style={align= center, rectangle, draw, double, rounded corners}]
\node  (n0) at (0,0) {Section \ref{sec:conditional_fooling}: Conditional Fooling Lemma (Lemma \ref{conditional_fooling}) \\ and Equidistribution Lemma (Lemma \ref{exponential_sum})};
\node  (n1) at (-20,-20) {Section \ref{sec:randomwalk}: PDT-Hardness $ \implies $ lower bounds \\ for depth-restricted $\text{Res}(\oplus)$ \\ (Theorem \ref{random_walk_with_restart_official})};
\node  (n2) at (20,-20) {Section \ref{sec:lifting_from_DT_hardness}: \\ DT-Hardness $\implies$ PDT-Hardness \\ (Theorem \ref{thm: lifting})};
\node  (n3) at (0,-40) {Theorem \ref{ref: DT_hardness_implies_lower_bounds}: \\ DT-Hardness $ \implies $ lower \\ bounds for depth-restricted $\text{Res}(\oplus)$};
\node  (n4) at (30,-40) {Section \ref{sec:proving_DT_hardness}: Tseitin contradiction \\ is DT-hard \\ (Theorem \ref{thm: DT_hard})};
\node  (n5) at (0,-60) {Section \ref{sec:final_result}: Proof of Main Theorem (\ref{tseitin_lower_bound}): \\ Lower bounds for near-quadratic depth $\text{Res}(\oplus)$ proofs};
\draw [->] (n0) -- (n1); 
\draw [->] (n0) -- (n2); 
\draw [->] (n1) -- (n3); 
\draw [->] (n2) -- (n3); 
\draw [->] (n3) -- (n5); 
\draw [->] (n4) -- (n5); 
\end{tikzpicture}

%% file: equidistribution_official.tex
Throughout this section, assume the gadget $g: \F^b \rightarrow \F$ has the following property.

\bigskip

\fbox{
\begin{minipage}{35em}
\begin{itemize}
\item For all $S \subseteq [b], |\hat{g}(S)| \leq \dfrac{1}{n^{2+\eta}}$ for some constant $\eta > 0$
\end{itemize}
\end{minipage}
} \vspace{2mm}

An explicit example of such a gadget is the Inner Product gadget on $(4+2\eta) \log(n)$ bits.

\bigskip

In this section, we will establish a key result that shows that lifted distributions \emph{fool} amortized closure (Lemma \ref{conditional_fooling}). 

In the following, we state a fact that will be, in some sense, a significant generalization of the following simple, well known fact: If $B\subseteq A\subseteq\mathbb{F}_2^{nb}$ are two affine spaces, then $Pr_{x \sim A}[x \in B] \le 2^{\text{codim}(A)-\text{codim}(B)}$. This fact was generalized recently by Efremenko and Itsykson \cite{EI25}. Let $y$ be an extendible assignment to the variables in closure of $A$, i.e. $\Cl(A)$. Then, Lemma 5.1 of \cite{EI25}, that they point out is their key lemma for improving the lower bound on resolution over parities, shows the following:
\begin{eqnarray}  \label{EI-fooling}
\Pr_{x \sim A \cap C_y}[x \in B] \le 2^{|\Clh(A)| - |\Clh(B)|}.
\end{eqnarray}

Note here we cannot hope to work with co-dimension of $B$ and $A$ as shown in the counter-example to the naive idea of conditional fooling as discussed earlier in Section~\ref{sec:intuition}. 
The argument in \cite{EI25} uses a convenient property of amortized closure, combined with simple linear algebra. In another direction, Bhattacharya, Chattopadhyay and Dvorak \cite{BCD24} showed the following: if $g$ is a gadget with certain properties, then the following is true for every $z \in \{0,1\}^n$:
\begin{eqnarray}  \label{BCD-fooling}
\Pr_{x \sim G^{-1}(z)}[x \in B] \le 2^{-\Omega(\text{codim}(B)/b)}
\end{eqnarray}

Below, we prove our main lemma which has the features  of both \eqref{EI-fooling} and \eqref{BCD-fooling}.

\begin{lemma}[Conditional Fooling Lemma]
\label{conditional_fooling}
    Let $B \subseteq A \subseteq \F^{nb}$ be affine subspaces such that $|\Clh(B)| = |\Clh(A)|+k$. Let $y \in \F^{\Cl(A) \times [b]}$ be an $A$-extendable closure assignment, and let $\mu$ be a distribution on $\F^n$ such that $z|_{\Cl(A)}=G(y)$ for every $z \in \text{supp}(\mu)$. Then,
    $$ \text{Pr}_{x \sim G^{-1}(\mu) \cap C_y} [x \in B | x \in A] \leq \left(\dfrac{3}{4}\right)^k$$
\end{lemma}

Currently, we do not know of a short argument to prove this. We prove it here in steps, establishing some equidistribution properties of gadgets with small Fourier coefficients wrt \emph{safe} affine spaces that seem independently interesting. 

\begin{lemma}[Equidistribution Lemma]
\label{exponential_sum}
Let $A\subseteq \F^{nb}$ be a safe affine space. Then, for all $z \in \F^{n},$ 
$$ \text{Pr}_{x \sim A} [G(x)=z] \in \left[1 \pm o(n^{-1-\eta/2}) \right] \dfrac{1}{2^n}$$
\end{lemma}

\begin{proof}

    Let $\text{codim}(A)=m$. Fix a $z \in \F^n$. Since $\text{Pr}_{x \leftarrow \F^{nb}} [x \in A]= 2^{-m},$ it suffices to show that
$$ \text{Pr}_{x \sim \mathbb{F}_2^{nb}} [x \in A \land G(x)=z] \in \left[ \dfrac{1-o(n^{-1-\eta/2})}{2^{m+n}}, \dfrac{1+o(n^{-1-\eta/2})}{2^{m+n}} \right].$$

    Let $M \in \F^{m \times nb}$ be a matrix for the equations defining $A$. Since $A$ is safe, there exist $m$ blocks such that one can choose one column from each block, such that those columns are linearly independent. WLOG (for notational convenience) assume those blocks are $1, 2, \cdots , m$, and from block $j$ we choose column $a_j$. \newline 

    We first rewrite the system of equations in a more convenient form. Since the matrix $M$ restricted to column set $S=\{(j,a_j) | 1 \leq j \leq m\}$ is invertible, we can perform row operations on $M$ so that the submatrix $ M_{[m], S}$ becomes $I_m$. Let $
    \ell_i$ denote the $i$-th row of this modified matrix. Thus, for every $i \in [m]$, there exists a $c \in [b]$ such that $l_i$ has a non-zero entry at coordinate $(i,c)$, and for every $i' \neq i$, $l_{i'}$ has a zero entry at coordinate $(i,c)$. An easy but crucial consequence of this is the following.
    
    \begin{observation} \label{crucial_observation} For every subset $T \subseteq [m]$, the vector $\displaystyle \sum_{j \in T} \ell_j$ has a non-zero coordinate in the $j$-th block for each $j \in T$.\end{observation} 

    Suppose the system of equations in this basis is
    \begin{align*}
        \langle \ell_1, x \rangle & = c_1 \\
        \langle \ell_2, x \rangle&  = c_2 \\ 
        \cdots & \cdots \\
        \langle \ell_m, x \rangle & = c_m
    \end{align*}

    \textit{Notation}: for an assignment $x \in \F^{nb}$, we denote by $x(i) \in \F^{b}$ the restriction of $x$ to the $i$'th block. \newline
    
    Let $p \coloneqq \text{Pr}_{x \sim \F^{nb}} [x \in A \land G(x)=z]$. We have

    \begin{align*}
        p & = \displaystyle \text{E}_{x} \left( \displaystyle \prod_{j=1}^{n} \left( \dfrac{1 + (-1)^{g(x(i))+z_i}}{2} \right) \displaystyle \prod_{j=1}^{m} \left( \dfrac{1 + (-1)^{\ell_j(x)+c_j}}{2} \right) \right)
    \end{align*}
    Expanding the RHS, we get the following expression:
    \begin{align*}
        p - \dfrac{1}{2^{n+m}} & = \displaystyle \sum_{\substack{S \subseteq[n]\\ T \subseteq [m] \\ S \cup T \neq \phi}} \text{E}_x \left[ \dfrac{(-1)^{\displaystyle \sum_{i \in S} (g(x(i))+z_i) + \displaystyle \sum_{j \in T} (\ell_j(x) + c_j)}}{2^{n+m}} \right]
    \end{align*}
    For $S \subseteq [n], T \subseteq [m]$ let $f_{S,T}(x) \coloneqq (-1)^{\displaystyle \sum_{i \in S} g(x(i)) + \displaystyle \sum_{j \in T} \ell_j(x)}$ and $u_{S,T} \coloneqq \displaystyle \sum_{i \in S} z_i + \displaystyle \sum_{j \in T} c_j$. We have

    \begin{align*}
    p - \dfrac{1}{2^{n+m}} &= \dfrac{1}{2^{n+m}}\displaystyle \sum_{\substack{S \subseteq [n] \\ T \subseteq [m] \\ S \cup T \neq \phi}} (-1)^{u_{S,T}} \displaystyle \text{E}_x [f_{S,T}(x)]
    \end{align*}

    We start by showing that $\text{E}_x [f_{S,T}(x)]$ vanishes unless $T \subseteq S$. This is where we use the safety of $A$. 
    \begin{claim}
    \label{vanishing}
        If $T \not \subseteq S$, $\text{E}_{x} [f_{S,T}(x)]=0$
    \end{claim}
    \begin{proof}
        Let $u \in T \setminus S$. By Observation \ref{crucial_observation}, there exists a coordinate $k$ in the $u$-th block on which $\displaystyle \sum_{j \in T} \ell_j$ is non-zero: $\displaystyle \sum_{j \in T} (\ell_j)_{(u,k)}=1$. Since $u \not \in S$, this coordinate does not affect $\displaystyle \sum_{i \in S} g(x(i))$. So we have that for all $x$, $f_{S,T}(x) = -f_{S,T}(x \oplus e_{u,k})$. Therefore, exactly half of the $x$'s have $f_{S,T}(x)=1$ and the result follows.
    \end{proof}

    It now suffices to bound the terms where $T \subseteq S$. We do this using the fact that all Fourier coefficients of $g$ are small.

    \begin{claim}
    \label{discrepancy}
        If $T \subseteq S$, $ | \text{E}_x [f_{S,T}(x)] | \leq  n^{-(2+\eta)|S|}.$
     \end{claim}
     \begin{proof}
       Let $g^{\oplus S}: \F^{b \times |S|} \rightarrow \F$ be the XOR of $|S|$ disjoint copies of $g$; for $y \in \F^{b \times |S|}$,
       $$ g^{\oplus S} (y) = \left( \displaystyle \sum_{i \in S} g(y(i))\right) \pmod{2}$$
       Note that $||\widehat{g^{\oplus S}}||_{\infty} = \left(||\hat{g}||_{\infty}\right)^{|S|} \leq n^{-(2+\eta)|S|}$. Therefore,
       $$ \left|  \text{E}_{x} [f_{S,T}(x)] \right| = \left| \widehat{g^{\oplus S}}\left(\text{supp}\left(\displaystyle \sum_{j \in T} l_j\right)\right)\right| \leq  n^{-(2+\eta)|S|}.$$
      \end{proof}

     Now, we upper bound the magnitude of the error as follows. If $|S|=k$, Claim~\ref{vanishing} implies that there are at most $2^k$ possible values of $T$ for which $\text{E}_x [f_{S,T}(x) ]\neq 0$. Claim \ref{discrepancy} implies that the magnitude of each of these terms is at most $ n^{-(2+\eta)k}$. Thus, we get that

    \begin{align*}
    \left| p - \dfrac{1}{2^{n+m}} \right| &= \dfrac{1}{2^{n+m}}\displaystyle \left|\sum_{\substack{S \subseteq [n] \\ T \subseteq [m] \\ S \cup T \neq \phi]}} (-1)^{u_{S,T}} \displaystyle \text{E}_x [f_{S,T}(x)] \right|
    \\ & \\
    \\ & \leq \dfrac{1}{2^{n+m}} \displaystyle \sum_{k=1}^{n} \binom{n}{k} 2^k  n^{-(2+\eta)k} \\
    & \leq \dfrac{1}{2^{n+m}} \displaystyle \sum_{k=1}^{n} \exp(k (\log (n) + 1 - (2 + \eta)\log(n))) \\
    & \leq \dfrac{1}{2^{n+m}} o(n^{-1-\eta/2})
    \end{align*}
This completes the proof.
    
\end{proof}

We will now show that the set of pre-images of an arbitrary $z \in \{0,1\}^n$, are approximately equidistributed among the various translates of a safe affine space in the lifted world. 

\begin{lemma}
    \label{uniform_coset}
    Let $A \subseteq \F^{nb}$ be a safe affine space with codimension $m$, and let $z \in \F^n$ be a target point. Then,
    $$ \text{Pr}_{x \sim G^{-1}(z)} [x \in A] \in \left[  \dfrac{1 - o(n^{-1-\eta/3})}{2^m}, \dfrac{1+o(n^{-1-\eta/3})}{2^m}\right]$$
\end{lemma}
\begin{proof}
    Let $A_1=A, A_2, \cdots, A_M$ be the $M= 2^m$ translates of $A$. Let $S_j= G^{-1}(z) \cap A_j$. Lemma \ref{exponential_sum} implies $\dfrac{|S_j|}{|A|} \in \left[ \dfrac{1-o(n^{-1-\eta})}{2^n}, \dfrac{1+o(n^{-1-\eta})}{2^n} \right]$ for all $j$. We have
    \begin{align*} \text{Pr}_{x \sim G^{-1}(z)} [x \in A]  &= \dfrac{|S_1|}{\displaystyle \sum_{j} |S_j|} \in \left[\dfrac{1-o(n^{-1-\eta/2})}{1+o(n^{-1-\eta/2})} \times \dfrac{1}{2^m} , \dfrac{1+o(n^{-1-\eta/2})}{1-o(n^{-1-\eta/2})} \times \dfrac{1}{2^m} \right] \\
    &= \left[ \dfrac{1-o(n^{-1-\eta/3})}{2^m}, \dfrac{1+o(n^{-1-\eta/3})}{2^m} \right]\end{align*}
    
\end{proof}

Using the above, we show below that if $B \subset A$ are two safe affine spaces, then $B$ cannot significantly distinguish the distributions $x \sim (G^{-1}(z) \cap A)$  and $x \sim A$.

\begin{lemma}
\label{fooling_nice_subspaces}
    Let $B \subseteq A \in \F^{nb}$ be safe affine subspaces such that $\text{codim}(B)= \text{codim}(A)+1$. Let $z \in \F^n$ be any point. Then,
    $$ \text{Pr}_{x \sim G^{-1}(z)} [x \in B | x \in A] \leq \dfrac{1}{2} + o(n^{-1-\eta/4})$$
\end{lemma}
\begin{proof}
    Let $m= \text{codim}(A)$. Lemma \ref{uniform_coset} implies $\text{Pr}_{x \sim G^{-1}(z)} [x \in A] \geq \dfrac{1-o(n^{-1-\eta/3})}{2^m}$ and $\text{Pr}_{x \sim G^{-1}(z)} [x \in B] \leq \dfrac{1+o(n^{-1-\eta/3})}{2^{m+1}}$
    Thus,

     $$ \text{Pr}_{x \sim G^{-1}(z)} [x \in B | x \in A] = \dfrac{ \text{Pr}_{x \sim G^{-1}(z)} [x \in B]}{\text{Pr}_{x \sim G^{-1}(z)} [x \in A]} \leq \dfrac{1+o(n^{-1-\eta/3})}{1-o(n^{-1-\eta/3})} \times \dfrac{1}{2} \leq \dfrac{1}{2} + o(n^{-1-\eta/4})$$
    
\end{proof}

The structure of safe affine spaces that we have discovered so far allows us to say the following about any two arbitrary affine spaces that are not necessarily safe.

\begin{lemma} \label{increase_by_one}
Let $B \subseteq A \in \F^{nb}$ be affine spaces such that $|\Clh(B)|= |\Clh(A)|+1$ and $\text{codim}(B)= \text{codim}(A)+1$. Let $y$ be an $A$-extendable closure assignment, and let $z \in \F^n$ be a point such that $z|_{\text{Cl}(A)}= G(y)$. Then,
$$ \text{Pr}_{x \sim G^{-1}(z) \cap C_y} [x \in B | x \in A] \leq \dfrac{1}{2}+o(n^{-1-\eta/4}) $$

\end{lemma}
\begin{proof}
    Let $z= (G(y), w)$. Rewrite the desired probability expression as
    $$ \text{Pr}_{\tilde{x} \sim G^{-1}(w)} [x \in B_y | x \in A_y]$$
    By Corollary \ref{both_are_nice}, $A_y, B_y$ are both safe affine subspaces, and $\text{codim}(B_y)= \text{codim}(A_y)+1$. Now the result follows from Lemma \ref{fooling_nice_subspaces}.
\end{proof}

An easy corollary is that the result still holds if we condition only on a subset of the blocks in $\Cl(A)$ instead of all the blocks in $\Cl(A)$.

\begin{corollary}
    \label{corollary: extended_version}
Let $B \subseteq A \in \F^{nb}$ be affine spaces such that $|\Clh(B)|= |\Clh(A)|+1$ and $\text{codim}(B)= \text{codim}(A)+1$. Let $S \subseteq \Cl(A)$ and let $y \in \F^{S \times [b]}$ be a partial assignment. Let $z \in \F^n$ be a point such that $z|_{S}= G(y)$ and $G^{-1}(z) \cap C_y \cap A \neq \phi$. Then,
$$ \text{Pr}_{x \sim G^{-1}(z) \cap C_y} [x \in B | x \in A] \leq \dfrac{1}{2}+o(n^{-1-\eta/4}) $$
\end{corollary}
\begin{proof}

  Sampling $x$ from $G^{-1}(z) \cap C_y$ can be done as follows: first sample $y^{(1)} \in \F^{\Cl(A) \times [b]} \cap C_y$ according to $G^{-1}(z)$, then sample $x$ from $G^{-1}(z) \cap C_{y^{(1)}}$. For each possible $y^{(1)}$ use Lemma \ref{increase_by_one} to upper bound the conditional probability of lying in $B$ conditioned on $y^{(1)}$. Formally, let $\mathcal{D}$ denote the distribution of $x|_{\Cl(A)}$ as $x \sim G^{-1}(z) \cap C_y$. Then,

    \begin{align*}
        \text{Pr}_{x \sim G^{-1}(z) \cap C_y} [x \in B | x \in A] & = \text{E}_{y^{(1)} \leftarrow \mathcal{D}} [\text{Pr}_{x \sim G^{-1}(z) \cap C_{y^{(1)}}} [x \in B | x \in A]] \\
        & \leq \dfrac{1}{2}+o(n^{-1-\eta/4})
    \end{align*}
  
\end{proof}

Now we prove the final result of this section.

\begin{proof} \emph{of Lemma~\ref{conditional_fooling}}
    By convexity, it suffices to prove the statement in the case that $\mu$ is concentrated on a single point. Thus, we have to show the following:

    \begin{mdframed}
        Let $B \subseteq A \subseteq \F^{nb}$ be affine spaces with $|\Clh(B)| \geq |\Clh(A)|+k$. Let $ y \in \F^{\Cl(A) \times [b]}$ be an $A$-extendable closure assignment. Let $z \in \F^n$ be a point such that $G(y)=z|_{\Cl(A)}$. Then,
        $$ \text{Pr}_{x \leftarrow G^{-1}(z) \cap C_y } [x \in B | x \in A] \leq \left(\dfrac{3}{4}\right)^k$$
    \end{mdframed}

    Let $B=W_0 \subseteq W_1 \subseteq W_2 \subseteq \cdots \subseteq W_{l-1} \subseteq W_l= A$ be a sequence of affine subspaces such that $\text{codim}(W_j)= \text{codim}(W_{j+1})+1$. 
    We have
    $$ \text{Pr}_{x \sim G^{-1}(z) \cap C_y} [x \in B | x \in A] = \displaystyle \prod_{j=0}^{l-1} \text{Pr}_{x \sim G^{-1}(z) \cap C_y} [x \in W_j | x \in W_{j+1}]$$

    We assume there exists a point in $G^{-1}(z) \cap C_y \cap B$ (as otherwise the conditional probability is 0), so in particular, for all $j$ there exists a point in $G^{-1}(z) \cap C_y \cap C_j$. \newline

    By Lemma \ref{lem: amortized_closure_continuity} there exist $k$ indices $j \in \{0, 1, \cdots , l-1\}$ such that $|\Clh(W_j)|= |\Clh(W_{j+1})|+1$. Note that $\Cl(A) \subseteq \Cl(W_{j+1})$ by Lemma \ref{closure containment}. Invoking Corollary \ref{corollary: extended_version} for each such index $j$, where $W_j$ plays the role of $B$ and $W_{j+1}$ that of $A$, we have 
    $$ \text{Pr}_{x \sim G^{-1}(z) \cap C_y} [x \in W_{j} | x \in W_{j+1} ] \leq \dfrac{1}{2} + o(1) \leq \dfrac{3}{4}$$
       So, in the product $\displaystyle \prod_{j=0}^{l-1} \text{Pr}_{x \sim G^{-1}(z) \cap C_y} [x \in W_j | x \in W_{j+1}]$, at least $k$ terms are $\leq  3/4$. The result follows.   
\end{proof}

%% file: CNF_description.tex
\subsection{Lifting CNFs}

\begin{definition}

For a base CNF $\Phi $ on variables $\{z_1,z_2, \cdots , z_n\}$ and a gadget $g: \F^n \rightarrow \F$ we define the lifted CNF $\Phi \circ g$ as follows.

\begin{itemize}
    \item The set of variables is $\{x_{i,j}  | i \in [n], j \in [b]\}$
    \item For each clause $C$ in $\Phi$, we define the set of clauses $C \circ g$ as follows: let the variables involved in $C$ be $\{x_i | i \in S\}$ and let $\alpha \in \F^{S}$ be the unique assignment to those variables that falsifies $C$. The set of clauses $C \circ g$ will involve variables from the set $\{x_{i,j} |i \in S, j \in [b]\} $. For every choice of $(a_i | i \in S)$ where $a_i \in g^{-1}(\alpha_i)$ we add the following clause to $C \circ g$:
    $$ \displaystyle \lor_{i \in S} \left[ \displaystyle \lor_{j \in [b]} [x_{i,j}  \neq \alpha_{i,j}] \right]$$
    \item The lifted CNF $\Phi \circ g$ is the conjunction of $C \circ g$ for every $C \in \Phi$.
\end{itemize}

The semantic interpretation of $\Phi \circ g$ is as follows: 
$$ \Phi \circ g(x) = \Phi(g(x_{1,1}, x_{1,2}, \cdots , x_{1,b}), g(x_{2,1}, x_{2,2} \cdots , x_{2,b}), \cdots , g(x_{n,1}, x_{n,2}, \cdots , x_{n,b}))$$
Thus if $\Phi$ is unsatisfiable, so is $\Phi \circ g$.

If the largest width of a clause in $\Phi$ is $w$ and $\Phi$ has $m$ clauses, the number of clauses in $\Phi \circ g$ will be at most $m 2^{bw}$. In particular, if $m \leq \text{poly}(n), b= O(\log(n))$ and $w= O(1)$ then the number of clauses of $\Phi \circ g$ is bounded by $\text{poly}(n)$. \end{definition}

\subsection{Choice of our base CNF}

The base CNF we shall use is the Tseitin contradiction over an expander graph, lifted with an appropriate gadget. Let $G= (V,E)$ be a $(|V|,d,\lambda < 0.95)$ expander with $|V|$ odd and $d= O(1)$. The base CNF $\Phi$ has variables $z_{u,v}$ for $(u,v) \in E$. For each $v \in V$ we express the constraint $\displaystyle \sum_{(v,w) \in E} z_{v,w} \equiv 1 \pmod{2}$ using $2^d= O(1)$ clauses. This system is unsatisfiable because adding up all the equations yields $0 \equiv 1 \pmod{2}$. \newline

The property of $G$ we shall use is isoperimetric expansion (which follows from Cheeger's inequality \cite{Che71}):

\fbox{
\begin{minipage}{35em}
    
\begin{lemma} \label{cheeger}
    For any $S \subseteq V$, the cut $E(S, V \setminus S)$ has at least $\dfrac{d}{40} \min(|S|, |V \setminus S|)$ edges.
\end{lemma}
\end{minipage}
} \newline

Explicit constructions of such graphs were provided in \cite{LPS88} and \cite{Mar73}. \newline

\subsection{Lifted CNF}
We lift $\Phi$ with an appropriate gadget. We will take the gadget $g: \F^b \rightarrow \F$ to be the Inner Product function $g=\text{IP}$ on $b= (4+\eta)\log(n)$bits for some arbitrarily small constant $\eta > 0$.
$$ \text{IP}(x_1, x_2, \cdots , x_{b/2}, y_1, y_2, \cdots,  y_{b/2}) = (x_1y_1 + \cdots + x_{b/2}y_{b/2}) \pmod{2}$$
 This satisfies the property mentioned in Section \ref{sec:conditional_fooling}: $||\hat{g}||_{\infty} \leq n^{-2-\eta/2}$.

\begin{theorem} \label{inner_product_fourier_norm}
For $g=\text{IP}$ on $b$ bits, $||\hat{g}||_{\infty} \leq 2^{-b/2}$
\end{theorem}

The CNF for which we will prove our depth restricted lower bounds is $\Psi = \Phi \circ g$. \newline

If we take $d=4$ we get graphs with $\lambda \leq 0.87 $. Then, the final resulting CNF has number of variables $N= O(m \log m)$, width $16 \log(m)$ and number of clauses $O(m^{17}) = \tilde{O}(N^{17})$. \newline

If we took $d=3$ instead, we would get number of clauses $\tilde{O}(N^{13})$. An issue with $d=3$ is that we are taking the right hand sides of each parity constraint to be 1, and $3$-regular graphs with $m$ vertices do not exist when $m$ is odd. One can get around this by defining general Tseitin contradictions where the right hand sides can be anything which sum up to $1 \pmod{2}$. The proof we present in our paper can be easily modified to work for general Tseitin contradictions. For simplicity we do not do this here.

%% file: random_walk_with_restarts_official.tex
Alekseev and Itsykson \cite{AI25} introduced the `random walk with restarts' approach to prove superlinear lower bounds on depth of $\text{Res}(\oplus)$ proofs of small size. To analyze their random walk with restarts, \cite{AI25} uses certain elaborate games. We find it more convenient to analyze random walks using the language of decision trees. In particular, this allows us naturally to bring in the notion of a hard distribution that seems crucial to boost the success probability of our random walk with restart significantly, all the way from $2^{-n/\log(n)}$ to a constant. In this section, we formalize our notion which we call $(p,q)$-PDT hardness. We point out that our notion here is a significant refinement of the ideas of Bhattacharya, Chattopadhyay and Dvor{\'a}k \cite{BCD24} where as well random walks on lifted distributions were analyzed, but without restarts.

We first set up some notation to define our hardness notion. For a parity decision tree $T$ and a point $x$, define the affine subspace $A_x(T)$ to be the one corresponding to the set of inputs $y$ that traverse the same path in $T$ as $x$ does. More formally, $A_x(T)$ is defined as follows: suppose on input $x$, $T$ queries the linear forms $\ell_1, \cdots , \ell_d$ and gets responses $c_1, c_2, \cdots , c_d$ respectively. Then, $A_x(T)= \{y | \langle \ell_j, y \rangle= c_j \forall j \in [d]\}$. 
    
We are ready now to introduce the notion of a hard set of partial assignments that will abstract our requirements for finding a deep node in the proof DAG.

\begin{definition} \label{def:random_walk_with_restarts}
Let $\Phi$ be a CNF formula on $n$ variables and $g: \F^b \rightarrow \F$ be a gadget. A set of partial assignments $P \subseteq \{0,1,*\}^n$ is $(p,q)$\textit{-PDT-hard} wrt $(\Phi,g)$ if the following properties hold:
\begin{itemize}
    \item \textbf{Non-emptiness: }$P \neq \phi$
    \item \textbf{No falsification: }No partial assignment in $P$ falsifies any clause of $\Phi$.
    \item \textbf{Downward closure: }If $\rho \in P$ and $\tilde{\rho}$ is obtained from $\rho$ by unfixing some of the bits set in $\rho$, then $\tilde{\rho} \in P$
    \item \textbf{Hardness against parity decision trees:} 
    Let $A \subseteq \F^{nb}$ be an affine space with $|\Clh(A)| \leq p$. Let $y \in \F^{\Cl(A) \times [b]}$ be an $A$-extendable closure assignment such that $\alpha= G(y) \in P$. Then, there exists a distribution $\mu$ on $\F^{n}$ such that the following properties hold:
    \begin{enumerate} 
        \item $z|_{\Cl(A)}= \alpha$ for all $z \in \text{supp}(\mu)$
        \item Let $T$ be any parity decision tree (with input $nb$ bits) of depth $\leq q$. For any $x$, define $\tilde{A}(x) = A_T(x) \cap A \cap C_y. $ With probability $\geq 1/3$, as  $x$ is sampled from $G^{-1}(\mu) \cap A \cap C_y$, it holds that $G(x)|_{\Cl(\tilde{A}(x))} \in P$ \footnote{$G^{-1}(\mu) \cap A \cap C_y$ is non-empty by Lemma \ref{exponential_sum} applied on the nice affine space $A_y$}. 
            
    \end{enumerate}
\end{itemize}
The pair $(\Phi,g)$ is $(p,q)$\textit{-PDT-hard} if it admits a $(p,q)$-PDT-hard set of partial assignments.
\end{definition}
We now state the main result of this section that shows $(p,q)$-PDT-hardness of a CNF is sufficient to get us good lower bound on depth of a refutation of the lifted formula, assuming the size of the refutation is small.

\begin{theorem}
    \label{random_walk_with_restart_official}
    Let $\Phi$ be a CNF on $n$ variables and let $g: \F^b \rightarrow \F$ be a gadget with $||\hat{g}||_{\infty} \leq n^{-2-\eta}$ for some $\eta > 0$. Suppose $(\Phi,g)$ is $(p,q)$-PDT-hard. Then, any $\reslin$ refutation of $\Phi \circ g$ of size $s$ must have depth at least $\Omega \left( \dfrac{pq}{\log(s)}\right)$.
\end{theorem}

To prove the above, we will first establish the following lemma. This lemma essentially tells us that as long as we are at a node whose associated affine space satisfies some convenient properties, we are assured to find another node at a distance $q$ from our starting node whose corresponding affine space continues to have reasonably convenient properties.  

\begin{lemma}
\label{get_next_node}
    Suppose $P$ is a set of partial assignments that is $(p,q)$-PDT-hard wrt $(\Phi,g)$ where $g: \F^b \rightarrow \F$ is a gadget with $||\hat{g}||_{\infty} \leq n^{-2-\eta}$ for some constant $\eta > 0$. Let $\Pi$ be a $\reslin$ refutation of $\Phi \circ g$  of size $s$. Let $v$ be a node in $\Pi$ such that $|\Clh(A_v)| \leq p$, and let $y \in \F^{\Cl(A_v) \times [b]}$ be an extendable closure assignment for $A_v$ such that $G(y) \in P$. Then, there exists another node $w$ in $\Pi$ such that:
    \begin{enumerate}
        \item There exists a length $q$ path from $v$ to $w$ in $\Pi$.
        \item There exists an extendable closure assignment for $A_w$, $\tilde{y}$, such that $G(\tilde{y}) \in P$.
        \item $|\Clh(A_w)| \leq |\Clh(A_v)| + 2 \log(s)$
    \end{enumerate}
\end{lemma}

\begin{proof}

    Let $\mu$ be the hard distribution guaranteed to exist by the definition of $(p,q)$-PDT-hardness.
    Let $T$ be the following parity decision tree: on any input $x$, it simulates the queries made by $\Pi$ starting from node $v$ for $q$ steps. For any $x \in A_v$, define $\text{END}_q(x)$ to be the node of $\Pi$ reached by $x$ starting from $v$ after $q$ steps. (In case $\Pi$ on $x$ reaches a leaf within $q$ steps starting from $v$, define $\text{END}_q(x)$ to be that leaf.) \newline 

    We have $A_T(x) \cap A_v \subseteq A_{\text{END}_q(x)}$. Let $\text{GOOD}= \{x | G(x|)_{\Cl(\tilde{A_v}(x))} \in P \}$ (recall, $\tilde{A_v}(x)= A_T(x) \cap A_v \cap C_y$). The definition of $(p,q)$-PDT-hardness guarantees that $\text{Pr}_{x \leftarrow G^{-1}(\mu) \cap A \cap C_y} [x \in \text{GOOD}] \geq 1/3$. \newline
    Let $\mathcal{N}= \{ \text{END}_q(x) | x \in \text{GOOD}\}$. Note that since no assignment in $P$ falsifies any clause of $\Phi$, no vertex in $\mathcal{N}$ is a leaf - and therefore, there is a length $q$ walk from $v$ to $w$ for all $w \in \mathcal{N}$ (i.e., the parity decision tree does not terminate before $q$ queries if $x \in \text{GOOD})$. Also, $A_T(x) \cap A_v \cap C_y \subseteq A_{\text{END}_q(x)}$, so $\Cl(A_{\text{END}_q(x))}) \subseteq \Cl(A_T(x) \cap A_v \cap C_y)$, so $x \in \text{GOOD}$ implies $G(x)|_{\Cl(\text{END}_q(x))} \in P$ (since $P$ is downward closed). Thus, properties (i) and (ii) are satisfied for all $w \in \mathcal{N}$. To complete the proof, we have to find a $w \in \mathcal{N}$ such that $|\Clh(A_w)| \leq |\Clh(A_v)| + 2 \log(s)$. \newline
    Since $|\mathcal{N}| \leq s,$ there exists a $w \in \mathcal{N}$ such that $\text{Pr}_{x \leftarrow G^{-1}(\mu) \cap C_y \cap A} [\text{END}_q(x)=w] \geq \dfrac{1}{3s}$. In particular, this implies
    $$ \text{Pr}_{x \leftarrow G^{-1}(\mu) \cap C_y } [x \in A_w | x \in A] \geq \dfrac{1}{3s}$$
    Lemma \ref{conditional_fooling} then implies $ |\Clh(A_w)| \leq |\Clh(A_v)| + 2 \log(s)$

\end{proof}

Now we are ready to prove our main result for this section, by repeatedly making use of Lemma~\ref{get_next_node}.

\begin{proof}[Proof of Theorem~\ref{random_walk_with_restart_official}]
    Let $\Pi$ be a $\reslin$ refutation of $\Phi \circ $. We shall inductively find vertices $v_1, v_2, \ldots,v_j$ in $\Pi$ for $j \leq \dfrac{p}{2 \log(s)}$ such that:
    \begin{itemize}
            \item $\text{depth}(v_j) \geq jq$
            \item $|\Clh(A_{v_j})| \leq 2j \log(s)$
            \item There exists an extendable closure assignment $y_j$ for $A_{v_j}$ such that $G(y_j) \in P$
    \end{itemize}
    For $j=0$ we pick the root. To get $v_{j+1}$ we apply Lemma \ref{get_next_node} to $v_j$. We can continue this way as long as  $|\Clh(A_{v_j})| \leq p$. Hence, we do this for $j = \left \lfloor \frac{p}{2\log s} \right \rfloor$ many steps.  In the end, we get a node at depth $\Omega \left( \dfrac{pq}{\log(s)} \right)$.
\end{proof}


%% file: defining_DT_hardness.tex
Note that $(p,q)$-PDT-hardness is a notion that measures the hardness of a set of partial assignments against parity decision trees. We define an analogous notion of hardness against ordinary decision trees - which we call DT hardness. We then exhibit a lifting theorem: a DT-hard formula is also PDT-hard. Given this, to establish lower bounds against depth-restricted $\text{Res}(\oplus)$, is suffices to argue against ordinary decision trees. \newline

\begin{definition} [$(p,q)$-DT hardness] \label{def: unlifted_hardness}
    For a CNF $\Phi$ on $n$ variables, call a set of partial assignments $P \subseteq \{0,1,*\}^n$ to be $(p,q)$-DT-hard if the following hold:
    \begin{itemize}
    \item \textbf{Non-emptiness: }$P \neq \phi$   
    \item \textbf{No falsification:} No partial assignment $\rho \in P$ falsifies any clause of $\Phi$.
    \item \textbf{Downward closure: } For any $\rho \in P$ and any $j \in [n]$, if $\tilde{\rho}$ is obtained by setting $\rho(j) \leftarrow *,$ then $\tilde{\rho} \in P$
    \item \textbf{Hard for decision trees: }For any $\rho \in P$ which fixes at most $p$ variables, there exists a distribution $\mu_{\rho}$ on the assignment to unfixed variables  such that the following holds:
    \begin{itemize}
        \item Let $T$ be a decision tree of depth $q$ querying the unfixed variables. If we sample an assignment to the unfixed variables from $\mu_{\rho}$ and run $T$ for $q$ steps, the partial assignment seen by the tree \footnote{The partial assignment seen by the tree is the partial assignment formed by bits queried by the tree and bits fixed by $\rho$.} also lies in $P$ with probability $\geq 1/2$.
    \end{itemize}
    \end{itemize}

    The CNF $\Phi$ is $(p,q)$-DT hard if it admits a set of $(p,q)$-DT-hard partial assignments.
\end{definition}

%% file: lifting.tex
The primary goal in this section is to establish a \textit{lifting theorem} from DT-hardness to PDT-hardness. The main result of this section is the following. \newline

\begin{theorem} \label{thm: lifting}
    If $\Phi$ is $(p,p+q)$-DT-hard, then $(\Phi,g) $ is $(p,q)$-PDT-hard where $g: \F^{b} \rightarrow \F$ is any gadget with $||\hat{g}||_{\infty} \leq n^{-2-\eta}$ for some constant $\eta > 0$ and $b= O(\log n)$.
\end{theorem}

\begin{remark} \label{lifting_stroner_than_needed}
    Note that in the definition of PDT hardness, the hard distribution for an affine space $A$ and an extendable closure assignment $y$ was allowed to depend on both $A$ and $y$. In our proof of Theorem~\ref{thm: lifting}, our hard distribution is independent of $A$ and depends only on $G(y)$. The hard distribution in the proof simply is the corresponding DT-Hard distribution.
\end{remark}
\subsection*{Conventions}

\begin{itemize}
    \item Throughout this section, we shall be analyzing the execution of parity decision trees on some input. Whenever we say a phrase like $\{ \text{intersection of first }t \text{ queries made by }T \},$ we mean the following affine subspace: let $\ell_1, \cdots , \ell_t$ be the first $t$ linear forms queried by $T$ and let $b_1, b_2, \cdots , b_r$ be the responses received. Then,
    $$ \{  \text{intersection of first }t \text{ queries made by }T \}\triangleq \{x | \langle l_j, x \rangle = b_j ]\text{ } \forall j \in [t]\}\}$$
    Whenever we talk about the \textit{current subspace} of a PDT $T$, we shall mean the subspace
    $$ A_{\text{safe}} \cap \{ \text{intersection of queries made by }T \text{ so far}\}$$
    (Here $A_{\text{safe}}$ is a safe affine subspace which will be clear from the context.)
    \item For $x \in \F^{nb}$ and $S \subseteq [n]$ let $\text{PROJ}(x,S) \subseteq \F^{S \times [b]}$ denote the projection of $x$ on the blocks in $S$. 
    \item For an affine space $A \subseteq \F^{nb}$ and $S \subseteq [n]$ define $A[\setminus S]$ to be the projection of $A$ on the blocks in $[n] \setminus S$, i.e. $A[\setminus S] \triangleq \{\text{PROJ}(x,[n]\setminus S)\, |\, x \in A\}$.
\end{itemize}

\subsection{High level overview}

Before proceeding, we give a high-level overview of our proof. Suppose $\Phi$ is $(p,p+q)$-DT-hard. Let $P \subseteq \{0,1,*\}^n$ be a hard set of assignments, and for each $\rho \in P$ with $|\rho| \leq p $ let $\mu_{\rho}$ be the hard distribution. We show the same family of hard distributions establishes $(p,q)$-PDT-hardness. Suppose this is not the case. Then, the following things exist:

\begin{mdframed}
    
\begin{itemize}
    \item A $\rho \in P$ with $|\rho| \leq p$. Let $\mu= \mu_{\rho}$ be its DT-hard-distribution.
    \item An affine subspace $A \subseteq \F^{nb}$ with $\Cl(A)= \text{fix}(\rho)$ and $|\Clh(A)| \leq p$.
    \item An $A$-extendable closure assignment $y \in \F^{\Cl(A) \times [b]}$ such that $G(y)= \rho$
    \item A PDT $T_{\text{par}}$ of depth $q$ such that as $x$ is sampled from $G^{-1}(\mu) \cap A \cap C_y$, with probability $\geq 2/3$, the following holds:
    \begin{itemize}
        \item Let $A(q)$ be the affine space seen by $T_{\text{par}}$ after $q$ queries. Then, $G(x)|_{\Cl(A(q))} \not \in P$
    \end{itemize}
\end{itemize}
\end{mdframed}

\vspace{2mm}
Given this, we attempt to construct an ordinary decision tree $T_{ord}$ of depth $p+q$ which contradicts the DT-hardness of the distribution $\mu_{\rho}$. Informally, we want that as $z$ is sampled from $\mu$, with probability $\geq 1/2$ the partial assignment seen by $T_{ord}$ (this includes the bits queried by $T_{ord}$ and bits fixed by $\rho$) does not lie in $P$. \newline

A first approach would be to try to use query-to-communication lifting results such as \cite{GPW17,CFKMP21} which show that one can simulate the execution of any communication protocol on a uniformly random preimage of $z$, with only a few queries to $z$ (where the $x$ variables in each $\text{IP}$ gadget belong to Alice and $y$ variables belong to Bob). Note that if the variables are distributed between Alice and Bob in any arbitrary fashion, a linear query can be processed using only one bit of communication - so low depth PDTs can be seen as low cost communication protocols. \newline

The issue is that the input to $T_{\text{par}}$ is not coming from $G^{-1}(z)$ - it is coming from $G^{-1}(z) \cap A \cap C_y$. The set $A \cap C_y$ is in general not a product set over the Alice and Bob variables- so it is unclear how to use known query-to-communication lifting results in a blackbox fashion here (it might be possible to modify the proof of \cite{GPW17} to work for this case -- but as we explain soon, this will not be needed). However, we have more structure: we are working with linear queries instead of general communication protocols, and moreover, $A \cap C_y$ projected on the blocks outside $\text{Cl}(A)$ is a safe subspace, so we have tools such as the Equidistribution Lemma in our hand. \newline

It turns out that the tools developed in Section \ref{sec:conditional_fooling} do indeed help us do a \cite{GPW17}-style simulation even when the input is sampled from $G^{-1}(z) \cap A \cap C_y$. Before formally describing the simulation, we provide an informal version. \newline

We are given a PDT $T_{\text{par}}$, and we want to simulate the execution of $T_{\text{par}}(x)$ when $x$ is sampled from $ G^{-1}(z) \cap A \cap C_y$, by making only a few queries to $z$. Let's make some straightforward simplifications: we can assume WLOG the linear queries are supported on blocks outside $\Cl(A)$, since the other variables are fixed. So we can assume the input is sampled from $G^{-1}(z') \cap A_y$ where $z'= z|_{[n] \setminus \Cl(A)}$. Denoting $A_y$ by $A_{\text{safe}}$ our goal becomes the following: \newline
\begin{mdframed}
Given an unknown input $z$, a safe affine space $A_{\text{safe}}$ and a PDT $T_{\text{par}}$, simulate the execution of $T_{\text{par}}(x)$ when $x$ is sampled from $G^{-1}(z) \cap A_{\text{safe}}$, by making a small number of queries to $z$. \newline
\end{mdframed}

First, we modify $T_{\text{par}}$ into a canonical form (which we call $\text{CANONIZE}(T_{\text{par}}, A_{\text{safe}})$ below). The execution of $\text{CANONIZE}(T_{\text{par}},A_{\text{safe}})$ proceeds as follows:
\begin{itemize}
    \item Suppose, $\ell$ is the linear form queried by $T_{\text{par}}$ at time $t$.
    \item Let, $A(t)$ be the affine subspace visited by $T_{\text{par}}$ at time $t$. That is,
    $$ A(t)= A_{\text{safe}} \cap \{ \text{intersection of first }t \text{ queries made by }T_{\text{par}}\}$$
    \item \textbf{Case 1: }Suppose, querying $\ell$ changes $\Cl(A(t-1))$ to $\text{CL}_{new}$. Then, $\text{CANONIZE}(T_{\text{par}},A_{\text{safe}})$ first queries all coordinates in the blocks of $\text{CL}_{new} \setminus \Cl(A(t-1))$ - and then it queries $\ell$    
    \item \textbf{Case 2: } Suppose querying $\ell$ does not change $\Cl(A(t-1))$ (note that this depends only on $\ell$ and not on the response received). Then, $\text{CANONIZE}(T_{\text{par}}, A_{\text{safe}})$ simply queries $\ell$.
    
\end{itemize}
Now the simulator tries to simulate $\text{CANONIZE}(T_{\text{par}},A_{\text{safe}})$ instead of $T_{\text{par}}$. The advantage of working with $\text{CANONIZE}(T_{\text{par}},A_{\text{safe}})$ is that it has a fixed closure assignment at any given point of execution. (It is easy to see that at any point of execution, the subspace of $\text{CANONIZE}(T_{\text{par}},A_{\text{safe}})$ is simply the subspace of $T_{\text{par}}$ with all the closure variables fixed.) Conditioned on this closure assignment, the subspace projected on the non-closure blocks is a safe affine subspace - and to analyze that, we have the machinery developed in Section \ref{sec:conditional_fooling} in our hands. \newline

Looking ahead, it turns out that the simulator does not need to query anything in Case 2. It simply samples a uniform random bit $b \in \F$ and feeds it to $T_{\text{par}}$. In Case 1, the simulator queries $z_i$ for all $i \in \text{CL}_{new}$ (except those that were previously queried). Now it has to feed a value $\beta \in \F^{\Cl_{new} \times [b]}$ to $T_{\text{par}}$ to continue the simulation. It has to generate this value in such a manner that the distribution of $\beta$ matches the distribution of the actual response if $x$ were sampled from $G^{-1}(z) \cap A(t-1)$. We shall show that this distribution is actually the same (upto small $\ell_1$ distance) for all $w$ such that $w_i=z_i \text{ }\forall i \in \text{CL}_{new}$. Thus, knowledge of the coordinates of $z$ in $\text{CL}_{new}$ suffices to generate $\beta$ accurately so that the simulation can proceed - the simulator simply picks an arbitrary $w$ which agrees with $z$ on $\text{CL}_{new}$ and samples from $G^{-1}(w) \cap [A(t-1)]$ (this is how the simulator saves on the total number of queries). Proving this requires tools developed in Section \ref{sec:conditional_fooling}. \newline

We describe the conversion of a PDT to its canonical form in more detail in Section \ref{subsec:canonical_PDT}. We prove our main lifting theorem in Section \ref{subsection:lifting}. Section \ref{subsec:aux} contains an auxiliary lemma required later in the proof; the reader is advised to skip this section until it is finally used (in Section \ref{subsection:lifting}). 

\subsection{Auxiliary Lemmata} \label{subsec:aux}

\begin{lemma} \label{lem:ratio_close_to_one}
Let $A \subseteq \F^{nb}$ be an affine space and let $S \subseteq [n]$ be a set of blocks such that $A[\setminus S]$ is safe. Let $g: \F^b \rightarrow \F$ be a gadget with $||\hat{g}||_{\infty} \leq n^{-2-\eta}$ for some constant $\eta>0$. Let $y \in \F^{S \times [b]}$ be an $A$-extendable assignment. Let $z,w \in \F^n$ such that $z_i=w_i= g(y(i)) \text{ } \forall i \in S$.

Then,
\begin{align} \label{eq:pre-image-size-ratio}
 \dfrac{|G^{-1}(z) \cap A \cap C_y|}{|G^{-1}(w) \cap A\cap C_y |} \in \left[ 1 - n^{-1-\eta/6}, 1 + n^{-1-\eta/6} \right]
\end{align}
\end{lemma}

\begin{proof}
    Let $z= (G(y), \tilde{z})$ and $w = (G(y), \tilde{w}). $ Then,
    $ |G^{-1}(z) \cap A \cap C_y| = |G^{-1}(\tilde{z}) \cap A_{y}| \in \left[ 1 \pm n^{-1-\eta/3} \right] \dfrac{|A_y|}{2^{n-|S|}}$, where the second inequality follows from the Equidistribution Lemma as $A_y$ is a safe affine space.
    We have similar bounds for $|G^{-1}(w) \cap A \cap C_y|$. Combining them, we get that the LHS of \eqref{eq:pre-image-size-ratio} is in $\left[ \dfrac{1-n^{-1-\eta/3}}{1+n^{-1-\eta/3}}, \dfrac{1+n^{-1-\eta/3}}{1-n^{-1-\eta/3}} \right] \subseteq [1-n^{-1-\eta/6}, 1+n^{-1-\eta/6}]$, as desired.

\end{proof}

\subsection{Canonical Parity Decision Trees} \label{subsec:canonical_PDT}

 \begin{definition} Let $A \subseteq \F^{nb}$ be an affine space. A set of blocks $S \subseteq [n]$ is $A$-\textit{compatible} if the following holds: let $y \in \F^{S \times [b]}$ be any $A$-extendable assignment. Then, $\text{Cl}(A \cap C_y)= S$. (This holds for one $A$-extendable assignment if and only this holds for all $A$-extendable assignments.) \end{definition}
    Informally, $S$ is $A$-compatible if and only if querying all the bits in $S$ changes $\Cl(A)$ to $S$. Observe that any $A$-compatible set must contain $\Cl(A)$.

 \begin{definition} Let $A \subseteq \F^{nb}$ be an affine space. A linear form $\ell$ is $A$-stationary if $\Cl(A)= \Cl(A \cap \{x | \langle \ell, x \rangle = 0 \})$
    \end{definition}
    Informally, $\ell$ is $A$-stationary if querying $\ell$ does not change the closure of $A$. \newline

    \vspace{2mm}

Let $A_{\text{safe}}\subseteq \F^{nb}$ be a \textit{safe} affine subspace.
\begin{definition}

An $A_{\text{safe}}$-canonical PDT $T$ is a PDT with the following properties: 

\begin{itemize}
    \item The execution of $T$ can be divided into \textit{phases}. At each phase, $T$ either executes a \textit{Type 1 query} or a \textit{Type 2 query}.
    \item Define $A(t) = A_{\text{safe}}\cap \{ \text{affine queries made by } T \text{ till Phase }t\}$
    \item \textbf{Type 1 query: }Select an $A(t-1)$-compatible subset of blocks $S \subseteq [n]$. Query all variables in the blocks in $S \setminus \Cl(A(t-1))$. 
    \item \textbf{Type 2 query: }Select an $A(t-1)$-stationary linear form $\ell$ and query $\ell$.
\end{itemize}

\end{definition}
Observe that in a canonical PDT, the affine space $A(t)$ fixes all bits in $\Cl(A(t))$. Thus, a canonical PDT maintains a \textit{closure assignment} at each phase. 

\begin{definition}
Let $A_{\text{safe}} \subseteq \F^{nb}$ be a safe affine space and let $T_{\text{par}}$ be a PDT. Define $\text{SPACE}(t) = A_{\text{safe}} \cap \{\text{first } t \text{ queries made by }T_{\text{par}} \}$. We define its canonized version, $\text{CANONIZE}(T_{\text{par}}, A_{\text{safe}})$ to be a canonical PDT which simulates $T_{\text{par}}$, whose description is given below: 
\newline

Throughout the execution, $\text{CANONIZE}(T_{\text{par}}, A_{\text{safe}})$ maintains a counter called $count$, starting at 0. It maintains the following invariant:

\begin{itemize}
    \item At every point of time, $\text{CANONIZE}(T_{\text{par}}, A_{\text{safe}})$ has queried the first $count$-many queries of $T_{\text{par}}$.
    \item Moreover, it holds that
    $$\text{Cl}\big(\text{current subspace of }\text{CANONIZE}(T_{\text{par}},A_{\text{safe}})\big) = \text{Cl}(\text{SPACE}(count))$$
    where current subspace of $\text{CANONIZE}(T_{\text{par}},A_{\text{safe}})$
$$         = A_{\text{safe}} \cap \{  \text{intersection of all queries made by CANONIZE}(T_{\text{par}},A_{\text{safe}}) \text{ so far}\} $$

\end{itemize}

We formally describe the execution of $\text{CANONIZE}(T_{\text{par}}, A_{\text{safe}})$ below. 

\begin{mdframed}
$\text{CANONIZE}(T_{\text{par}}, A_{\text{safe}})$
\hrule \vspace{1mm}
\textbf{Input}: Query access to $x \in \F^{nb}$ \vspace{1mm}
\hrule \vspace{1mm}
\textbf{Internal variables:} \vspace{1mm}
\begin{itemize}
    \item $count$= number of queries of $T_{\text{par}}$ executed so far. Initialized to 0.
    \item $v$: Node of $T_{\text{par}}$ the tree is currently executing
    
\end{itemize}
\hrule
\vspace{1mm}
\textbf{Execution}

$\text{While }v \text{ is not a leaf of }T_{\text{par}}:$
\hspace{3mm} 
\begin{itemize}
    \item \textbf{Let:}
    \begin{itemize}
        \item $A_{cur}= \text{current affine subspace of }\text{CANONIZE}(T_{\text{par}},A_{\text{safe}})$
        \item $S_{cur}= \text{Cl}(A_{cur})= \text{Cl}(\text{SPACE}(count))$
        \item $\beta = \text{assignment to }S \text{ fixed by  } A_{cur}$ 
        \item $\ell= \text{linear form queried at }v$ by $T_{\text{par}}$
        \end{itemize}
        \item \textbf{Case 1: }$\ell $ is determined by $A_{cur}$ 
        \begin{itemize}
            \item Increase $count$ by 1.
            \item Let $b= \text{fixed value of }\ell$. Update $v$ according to $b$.
        \end{itemize}
        \item \textbf{Case 2: }$\ell$ is a $A_{cur}$-stationary linear form 
        \begin{itemize}
            \item Start a new phase.
            \item Execute a \textit{type 2 query} by querying $\ell$. Let $b$ be the response received. 
            \item Increase $\text{counter}$ by 1.
            \item Update $v$ according to $b$.
        \end{itemize}
        \item \textbf{Case 3: } Querying $\ell$ changes $\text{Cl}(A_{cur})$
        \begin{itemize}
            \item Start a new phase.
            \item Let $S_{new}$ be the new closure. Execute a \textit{type 1 query} by querying all coordinates in the blocks of $S_{new} \setminus S_{cur}$. 
            \item Start a new phase. 
            \item Execute a \textit{type 2 query} by querying $\ell$. Let $b$ be the response received.
            \item Increase $\text{counter}$ by 1.
            \item Update $v$ according to $b$.
            
    \end{itemize}
\end{itemize}
\end{mdframed}

\end{definition}

Let $T$ be any canonical PDT. We define $\text{TRANSCRIPT}(x,T,t)$ to be the transcript of $T$ after $t$ phases on input $x$. This consists of the following data: \newline

\begin{mdframed}
\begin{itemize}
    \item $v(t)= \text{vertex of }T \text{ reached after }t \text{ phases}$
    \item $A(t)= A_{\text{safe}} \cap \{  \text{intersection of queries made during first }t \text{ phases}\}$
    \item $S_t= \Cl(A(t))$ 
    \item $\beta(t)=$ fixed assignment to $S_t$. (Recall that $A(t)$ fixes all bits in $S_t$.) 
    \item $ \text{LIN-QUERY}(t) = \{  (\ell_1, b_1), (\ell_2, b_2), \cdots , (\ell_r, b_r) \}$ is the set of type 2 queries made by $T$ during first $t$ phases and the responses received. 
   
\end{itemize}
\end{mdframed} \vspace{2mm}

Technically, this above data is redundant because $A(t), S_t, \beta(t), \text{LIN-QUERY}(t)$ are all determined by $v(t)$. But we still describe the transcript in this manner because it facilitates describing how the transcript gets updated at each phase. Let the transcript of $T$ on $x$ be the transcript generated at the unique leaf $T$ reaches on $x$.\newline

\subsection{Proof of Lifting Theorem} \label{subsection:lifting}

In the following, we are given a canonical PDT $T$ which arises as the \textit{canonization} of another PDT $T_{\text{par}}$. We construct an ordinary randomized DT $T_{\text{ord}}$, which, on input $z$, outputs a transcript of $T_{\text{par}}$. Our goal is the following: over the internal coin tosses of $T_{\text{ord}}$, the output distribution of $T_{\text{ord}}(z)$ should be close in statistical distance to the transcript distribution of $T_{\text{par}}(x)$ when $x$ is sampled from $G^{-1}(z)\cap A_{\text{safe}}$. This theorem is inspired from \cite{GPW17}. We formalize the statement below.  

\begin{theorem} \label{thm: simulation}
Let $A_{\text{safe}} \subseteq \F^{nb}$ be a safe affine space with $|\Clh(A_{\text{safe}})| \leq p$. Let $T_{\text{par}}$ be a PDT with depth $\leq q$ and $T= \text{CANONIZE}(T_{\text{par}}, A_{\text{safe}})$. Let $g: \F^b \rightarrow \F$ be a gadget with $||\hat{g}||_{\infty} \leq n^{-2-\eta} $ for some constant $\eta>0$. There exists an ordinary randomized decision tree $T_{\text{ord}}$ of depth $\leq p+q$ such that the following two distributions have statistical distance $\leq b \cdot n^{-\eta/20}$ for all $z$:
\begin{itemize}
    \item \textbf{Distribution 1: }Sample $x \leftarrow G^{-1}(z) \cap A_{\text{safe}}$. Output \fbox{$\text{Transcript of } T \text{ on }x$} 
    \item \textbf{Distribution 2: }Output $T_{\text{ord}}(z)$
\end{itemize}
\end{theorem}

\begin{proof}

Throughout this proof, we denote by $A(t)$ the subspace after first $t$ phases of $T$, i.e.
$$ A(t) = A_{\text{safe}} \cap \{  \text{intersection of queries made by }T \text{ in first }t \text{ phases}\}$$.

The decision tree $T_{\text{ord}}$ runs the simulation routine $\text{SIMULATE}(T,A_{\text{safe}},z)$ as described below.
\newline

\begin{mdframed}
$\text{SIMULATE}(T, A_{\text{safe}}, z)$

\hrule
\vspace{1.5mm}
\textbf{Input}: 
\begin{itemize}
    \item Safe affine space $A_{\text{safe}} \subseteq \F^{nb}$
    \item $A_{\text{safe}}$-canonical PDT $T$
    \item Query access to a point $z \in \F^n$
\end{itemize}

\hrule
\vspace{1.5mm}
\textbf{Algorithm}
\begin{itemize}
    \item Define $\text{TR}_0 = \text{initial transcript}, v(0):$
    \begin{itemize}
        \item $A(0)= A_{\text{safe}}$
        \item $S_0= \phi$
        \item $\beta(0)= \text{empty assignment} $
        \item $\text{LIN-QUERY}(0)= \phi$
        \item $v(0)= \text{root of }T$
    \end{itemize}
    \item For $t= 1, 2, \cdots , d= \text{depth}(T)$,
    $$ \text{TR}_t = \text{UPDATE-SIMULATION}(T, \text{TR}_{t-1}, A_{\text{safe}}, t, z)$$
    \item Return $\text{TR}_d$
\end{itemize}

\vspace{1.5mm}

\end{mdframed} \vspace{2mm}

Here $\text{UDPATE-SIMULATION}(T,\text{TR}_{t-1},A_{\text{safe}},t,z)$ is a randomized algorithm which takes as input a transcript at phase $t-1$ and outputs a transcript at phase $t$. This is implemented as follows: \newline

\begin{mdframed}
$\text{UPDATE-SIMULATION}(T, \text{TR}_{t-1}, A_{\text{safe}}, t, z)$
\hrule
\vspace{1.5mm}
\textbf{Input}:
\begin{itemize}
    \item $A_{\text{safe}} \subseteq \F^{nb}$ safe affine space
    \item $T$: $A_{\text{safe}}$-canonical PDT
    \item $\text{TR}_{t-1}$: transcript at phase $t-1$
    \item Query access to a point $z \in \F^n$
\end{itemize}
\hrule
\vspace{1.5mm}
\textbf{Output}:

\begin{itemize}
    \item $\text{TR}_t$: transcript at phase $t$
\end{itemize}
\hrule
\vspace{1.5mm}
\textbf{Algorithm:}
\begin{itemize}
    \item \textbf{Case 0: }$v(t-1)$ is a leaf. 
    \begin{itemize}
        \item Do nothing. Return input transcript.
    \end{itemize}

    \item \textbf{Case 1: }$T$ makes a type 1 query at $v(t-1)$.

    \begin{itemize}
        \item Update $S_{t-1}$ to $S_t$. Sample the new closure assignment $\beta(t)$ as follows:
        \begin{itemize}
        \item Query $z_j$ for $j \in S_{t} \setminus S_{t-1}$
        \item Pick an arbitrary $w$ from the set $\{w | w_j = z_j \text{ for all }j \in S_t\}$ \footnote{All bits in $S_t$ have been queried at this point.}
        \item Sample $x \leftarrow G^{-1}(w) \cap A(t-1)$ uniformly at random \footnote{Since $A(t-1)$ fixes all bits of $S_{t-1}$, this automatically ensures $x$ is consistent with $\beta(t-1)$}
         \footnote{Since $S_t$ is $A(t-1)$-compatible, the subspace $A(t-1)[\setminus S]$ is safe, so the set in question is non-empty by the equidistribution lemma.}
        \item Set $\beta(t)= \text{PROJ}(x, S_t)$
        \end{itemize}
        \item Update $v(t)$ according to $\beta(t)$
    \end{itemize}
    \item \textbf{Case 2: }$T$ makes a type 2 query $\ell$ at $v(t-1)$. 
    \begin{itemize}
        \item Sample a bit $b \in \F$ uniformly at random.
        \item Update $\text{LIN-QUERY}(t) = \text{LIN-QUERY}(t-1) \cup \{(\ell ,b)\} $
        \item Update $v(t)$ according to $b$
    \end{itemize}
\end{itemize}
\end{mdframed}
\vspace{2mm}Since $|\Clh(A_{\text{safe}})| \leq p$, the size of the amortized closure can increase by at most 1 at each query and the closure is a subset of the amortized closure, the final size of the closure of the affine space reached by $T_{\text{par}}$ (and therefore that of $T= \text{CANONIZE}(T_{\text{par}}, A_{\text{safe}})$) \footnote{To clarify, the affine space reached by $T_{\text{par}}$ is $A_{\text{safe}} \cap \{ \text{intersection of all queries made by }T_{\text{par}}\}$.} is at most $p+q$. Thus, $T_{\text{ord}}$ makes at most $p+q$ queries. \newline

It remains to show that the resulting distribution is close to $\text{TRANSCRIPT} (T, G^{-1}(z)) \cap A_{\text{safe}})$. To do so, we shall employ a hybrid argument. \newline

We define a procedure for generating $\text{TRANSCRIPT}(T, x)$ when $x$ is sampled uniformly at random from $G^{-1}(z) \cap A_{\text{safe}}$. We call this procedure $\text{GENERATE}(T, A ,z)$ (this procedure assumes full knowledge of $z$). We shall be comparing the output distributions of $\text{GENERATE}(T,A_{\text{safe}},z)$ and $\text{SIMULATE}(T,A_{\text{safe}},z)$.

\begin{mdframed}
$\text{GENERATE}(T, A_{\text{safe}}, z)$

\hrule
\vspace{1.5mm}
\textbf{Input}: 
\begin{itemize}
    \item Safe affine space $A_{\text{safe}} \subseteq \F^{nb}$
    \item $A_{\text{safe}}$-canonical PDT $T$
    \item A point $z \in \F^n$ (with full access)
\end{itemize}

\hrule
\vspace{1.5mm}
\textbf{Algorithm}
\begin{itemize}
    \item Define $\text{TR}_0 = \text{initial transcript}:$
    \begin{itemize}
        \item $A(0)= A_{\text{safe}}$
        \item $S_0= \phi$
        \item $\beta(0)= \text{empty assignment} $
        \item $\text{LIN-QUERY}(0)= \phi$
        \item $v(0)= \text{root of }T$
    \end{itemize}
    \item For $t= 1, 2, \cdots , d= \text{depth}(T)$,
    $$ \text{TR}_t = \text{UPDATE-ACTUAL}(T, \text{TR}_{t-1}, A_{\text{safe}}, t, z)$$
    \item Return $\text{TR}_d$
\end{itemize}

\hrule
\vspace{1.5mm}

\end{mdframed} \vspace{2mm}

\vspace{2mm}
Here $\text{UPDATE-ACTUAL}(T, \text{TR}_{t-1}, A_{\text{safe}}, t, z)$ is the function which updates the transcript according to the appropriate conditional distrbution. This is implemented as follows:

\begin{mdframed}
$\text{UPDATE-ACTUAL}(T, \text{TR}_{t-1}, A_{\text{safe}}, t, z)$
\hrule
\vspace{1.5mm}
\textbf{Input}:
\begin{itemize}
    \item $A_{\text{safe}} \subseteq \F^{nb}$ safe affine space
    \item $T$: $A_{\text{safe}}$-canonical PDT
    \item $\text{TR}_{t-1}$: a transcript at phase $t-1$
    \item A point $z \in \F^n$ (with full access)
\end{itemize}
\hrule
\vspace{1.5mm}
\textbf{Output}:

\begin{itemize}
    \item $\text{TR}_t$: a transcript at phase $t$
\end{itemize}
\hrule
\vspace{1.5mm}
\textbf{Algorithm:}
\begin{itemize}
    \item \textbf{Case 0: }$v(t-1)$ is a leaf. 
    \begin{itemize}
        \item Do nothing. Return input transcript.
    \end{itemize}
    \item \textbf{Case 1: }$T$ makes a type 1 query at $v(t-1)$. \newline
    In $\text{TR}_t$, $S_{t-1}$ gets replaced by $S_t$.
    \begin{itemize}
        \item Replace $S_{t-1}$ by $S_t$. Sample the new closure assignment $\beta(t)$ as follows:
        \begin{itemize}
        \item Sample $x$ uniformly at random from $G^{-1}(z) \cap A(t-1)$ 
        \item Set $\beta(t) = \text{PROJ}(x, S_t)$ for all $j \in S_t$ \footnote{Since $A(t-1)$ fixes all bits in $S_{t-1}$ to $\beta(t-1)$, this automatically ensures $\beta(t)$ is consistent with $\beta(t-1)$. Also, by equidistribution lemma, $A(t-1) \cap G^{-1}(z)$ is non-empty (although that is not required because if we are following a valid execution transcript, this set is guaranteed to be non-empty regardless of the gadget).}
        \end{itemize}
        \item Update $v(t)$ according to $\beta(t)$
    \end{itemize}
    \item \textbf{Case 2: }$T$ makes a type 2 query $\ell$ at $v(t-1)$.
    \begin{itemize}
        \item Let $\lambda = \text{Pr}_{x \in A(t-1) \cap G^{-1}(z)}[ \langle \ell, x \rangle = 0]$.
        \item Let $b$ be a random binary variable which is 0 with probability $\lambda$, 1 with probability $1-\lambda$.
        \item $\text{LIN-QUERY}(t) = \text{LIN-QUERY}(t-1) \cup \{ (\ell, b) \}$
        \item Update $v(t)$ according to $b$
    \end{itemize}
\end{itemize}
\end{mdframed}
\vspace{2mm}

Our goal is to show the output distributions of $\text{SIMULATE}(T,A_{\text{safe}},z)$ and $\text{GENERATE}(T,A_{\text{safe}},z)$ are close in statistical distance.
A transcript from $\text{GENERATE}(T,A_{\text{safe}},z)$ is generated by the following sampling procedure. \newline

\begin{mdframed}
\begin{itemize}
    \item Start with the initial transcript $\text{TR}_0$
    \item For $t= 1, 2, \cdots , d= \text{depth}(T)$,
    $$\text{TR}_t \leftarrow \text{UPDATE-ACTUAL}(T, \text{TR}_{t-1}, A_{\text{safe}}, t, z)$$ 
    \item Output $\text{TR}_d$.
\end{itemize}
\end{mdframed}
\vspace{2mm} 
A transcript from $\text{SIMULATE}(T,A_{\text{safe}},z)$ is generated by the following sampling procedure. \newline

\begin{mdframed}
\begin{itemize}
    \item Start with the initial transcript $\text{TR}_0$
    \item For $t= 1, 2, \cdots , d= \text{depth}(T)$, 
    $$\text{TR}_t \leftarrow \text{UPDATE-SIMULATION}(T, \text{TR}_{t-1}, A_{\text{safe}}, t, z)$$ 
    \item Output $\text{TR}_d$.
\end{itemize}
\end{mdframed}
\vspace{2mm}

We define a sequence of hybrids $\mathcal{H}_0, \mathcal{H}_1, \cdots , \mathcal{H}_{d}$ ($d= \text{depth}(T)$) such that $\mathcal{H}_d = \text{GENERATE}(T,A_{\text{safe}},z)$  and $\mathcal{H}_0= \text{SIMULATE}(T,A_{\text{safe}},z)$. \vspace{1mm}

\begin{mdframed}
    
\textbf{Hybrid $\mathcal{H}_i$}
\begin{itemize}
    \item Start with the initial transcript $\text{TR}_0$
    \item For $t= 1, 2, \cdots , d$, 
    $$\text{TR}_t \leftarrow \begin{cases}
        \text{UPDATE-ACTUAL}(T, \text{TR}_{t-1}, A_{\text{safe}}, t, z) & \text{ if }t < i \\
        \text{UPDATE-SIMULATION}(T, \text{TR}_{t-1}, A_{\text{safe}}, t, z) & \text{ if } t\geq i
    \end{cases}$$ 
    \item Output $\text{TR}_d$.
\end{itemize}
\end{mdframed}\vspace{2mm}
To show the statistical distance between the outputs of $\mathcal{H}_0, \mathcal{H}_d$ is small, we show that the statistical distance between the outputs of $\mathcal{H}_i$ and $\mathcal{H}_{i+1}$ is small for all $i$. More specifically, we shall show that $d_{TV}(\mathcal{H}_i, \mathcal{H}_{i+1}) \leq o(n^{-1-\eta/20}))$. Since $d \leq nb$, this implies that the total simulation error is at most $o(b \cdot n^{-\eta/20})$, as claimed.
\newline

Since $\mathcal{H}_i$ and $\mathcal{H}_{i+1}$ differ only in a single update step, it suffices to show that changing just one update step from $\text{UPDATE-ACTUAL}$ to $\text{UPDATE-SIMULATION}$ does not significantly affect the distribution (in statistical distance). Since at phases $i+1, i+2, \cdots ,$ $\mathcal{H}_i$ and $\mathcal{H}_{i+1}$ get identically updated, it suffices to show the distributions of the new data added (in case of a type 1 query, the partial assignment $\beta(t)$; in case of a type 2 query, the bit $b$) are almost identically distributed by $\text{UPDATE-ACTUAL}$ and $\text{UPDATE-SIMULATION}$. \newline

We look at both possible types of queries and analyze the resulting distributions. \newline

\begin{mdframed}
\textbf{When $T$ makes a type 1 query}
\hrule
\vspace{1.5mm}
\textit{Common data}
\begin{itemize}
    \item Sets $S_{t-1}$ (closure at phase $t-1$) and $S_t$ (new closure)
    \item $\beta(t-1) :$ assignment to variables in blocks of $S_{t-1}$
    \item $A(t-1)$: affine space seen by $T$ at phase $t-1$ 
\end{itemize}
\hrule 
\vspace{1.5mm}
\textit{Output}
\begin{itemize}
    \item $\beta(t)$: assignments to variables in blocks of $S_t$
\end{itemize}
\hrule
\vspace{1.5mm}
\textit{Sampling procedure of $\text{UPDATE-ACTUAL}$} 
\begin{itemize}
    \item Sample $x \in A(t-1) \cap G^{-1}(z)$ uniformly at random.
    \item Set $\beta(t)= \text{PROJ}(x, S_t)$
\end{itemize}

\hrule
\vspace{1.5mm}

\textit{Sampling procedure of $\text{UPDATE-SIMULATION}$}
\begin{itemize}
    \item Query $z_j$ for $j \in S_t$
    \item Pick an arbitrary $w \in \F^{n}$ from the set $\{w | w_i = z_i \forall i \in S_t\}$.
    \item Sample $x \in A(t-1) \cap G^{-1}(w)$ uniformly at random.
    \item Set $\beta(t)= \text{PROJ}(x, S_t)$.
\end{itemize}
\hrule
\vspace{1.5mm}

\textit{Showing these distributions are close} \newline

\vspace{1.5mm}
We shall show that for all $w \in \F^n$ such that $w_i=z_i \text{ } \forall i \in S_t$, the distributions $\nu_1= \text{PROJ}(G^{-1}(w) \cap A(t-1), S_t)$ and $\nu_2= \text{PROJ}(G^{-1}(z) \cap A(t-1), S_t)$ are $n^{-1-\eta/8}$ close to each other. This implies the result. \newline

We shall show that for all $y \in \{0,1\}^{S_t \times [b]}$ such that $y$ is consistent with $\beta(t-1)$ and $g(y(i))=z_i=w_i \text{ } \forall i \in S_t$, it holds that $\text{Pr}_{\nu_1}[y] \in [1 \pm o(n^{-1-\eta/8})] \text{Pr}_{\nu_2}[y]$. (It is easy to see all $y \in \text{supp}(\nu_1) \cup \text{supp}(\nu_2)$ satisfy these two conditions.)\newline

We have
$$ \text{Pr}_{\nu_1}[y]= \dfrac{| G^{-1}(w) \cap A(t-1) \cap C_y|}{| G^{-1}(w) \cap A(t-1) \cap C_{\beta(t-1)} |}$$

and a similar expression for $\text{Pr}_{\nu_2}[y]$. 
Thus,

$$ \dfrac{\text{Pr}_{\nu_1}[y]}{\text{Pr}_{\nu_2}[y]}= \dfrac{| G^{-1}(w) \cap A(t-1) \cap C_y|}{| G^{-1}(z) \cap A(t-1) \cap C_y|} \times \dfrac{| G^{-1}(z) \cap A(t-1) \cap C_{\beta(t-1)} |}{| G^{-1}(w) \cap A(t-1) \cap C_{\beta(t-1)} |} $$

Since $S_t$ is an $A(t-1)$-compatible set, we get that $A(t-1)[\setminus S_t]$ is a safe space. Applying Lemma \ref{lem:ratio_close_to_one} to $A(t-1)$ and $\beta(t)$, the first term lies in $[1-n^{-1-\eta/6}, 1+n^{-1-\eta/6}]$. Since $\beta(t-1)$ fixes precisely the lifted variables in $\text{Cl}(A(t-1))$, we can also apply Lemma \ref{lem:ratio_close_to_one} on $A(t-1)$ and $\beta(t-1)$, which gives us that the second term lies in $[1-n^{-1-\eta/6}, 1+n^{-1-\eta/6}]$. Thus,
$\text{Pr}_{\nu_1}[y]/\text{Pr}_{\nu_2}[y] \in [1 - n^{-1-\eta/8}, 1+n^{-1-\eta/8}]$. Finally, we have

\begin{align*}
    d_{TV}(\nu_1, \nu_2) & = \dfrac{1}{2} \displaystyle \sum_{y} |\text{Pr}_{\nu_1}(y) - \text{Pr}_{\nu_2}(y)| \\
    & \leq \dfrac{1}{2} \times n^{-1-\eta/8} \displaystyle \sum_{y} \text{Pr}_{\nu_1}(y) \\
    & < n^{-1-\eta/8}
\end{align*}
as desired.

\end{mdframed}

\begin{mdframed}
\textbf{When $T$ makes a type 2 query}
\hrule
\vspace{1.5mm}
\textit{Common data}
\begin{itemize}
    \item $S_{t-1}$ (closure at phase $t-1$)
    \item $\beta(t-1) :$ assignment to variables in blocks of $S_{t-1}$
    \item $A(t-1)$: affine space seen by $T$ at phase $t-1$ 
    \item $\ell$: the new linear form queried
\end{itemize}
\hrule 
\vspace{1.5mm}
\textit{Output}
\begin{itemize}
    \item A bit $b \in \{0,1\}$
\end{itemize}
\hrule
\vspace{1.5mm}
\textit{Sampling procedure of $\text{UPDATE-ACTUAL}$} 
\begin{itemize}
    \item Let $\lambda= \text{Pr}_{x \leftarrow G^{-1}(z) \cap A(t-1)} [\langle \ell , x \rangle = 0] $
    \item 
    $$ \text{Output }b= \begin{cases}
        0 & \text{with probability }\lambda \\
        1 & \text{with probability }1-\lambda  \end{cases}$$
\end{itemize}

\hrule
\vspace{1.5mm}

\textit{Sampling procedure of $\text{UPDATE-SIMULATION}$}
\begin{itemize}
    \item 
    $$ \text{ Output }b= \begin{cases}
        0 & \text{with probability }1/2 \\
        1 & \text{with probability }1/2
    \end{cases}$$
\end{itemize}
\hrule
\vspace{1.5mm}

\textit{Showing these distributions are close} \newline

\vspace{1.5mm}
To show that the distributions of $b$ are close in both cases, we need to establish that $\lambda= \text{Pr}_{x \leftarrow G^{-1}(z) \cap A(t-1)} [\langle \ell , x \rangle = 0] $ is close to $1/2$. Let $A' = A(t-1) \cap \{x | \langle \ell, x \rangle = 0 \}$. Since $\ell$ is a type 2 query, $\text{Cl}(A') = \text{Cl}(A(t-1)) = S_{t-1}$. Our goal is to show
$$ \text{Pr}_{x \leftarrow G^{-1}(z)} [x \in A' | x \in A(t-1)] \in  \left[ \dfrac{1}{2} \pm o(n^{-1-\eta/6}) \right]$$
Let $z = (z_1, z_2)$ where $z_1 \in \F^{S_t}, z_2 \in \F^{[n] \setminus S_t}$. Rewrite the above expression as
$$ \text{Pr}_{u \leftarrow G^{-1}(z_2)} [u \in (A')_{\beta(t-1)} | u \in A(t-1)|_{\beta (t-1)}]$$
Since $\text{Cl}(A') = \text{Cl}(A(t-1))= S_{t-1}$, both the affine spaces $(A')_{\beta(t-1)}$ and $(A(t-1))_{\beta(t-1)}$ are safe. Let $k= \text{codim}\big((A(t-1))_{\beta(t-1)}\big)$. Note that $\text{codim}((A')_{\beta (t-1)})= k+1$ \footnote{We have $\text{codim}( (A')_{\beta (t-1)}) \in \{k, k+1\}$ . If $\text{codim}((A')_{\beta (t-1)})$ were $k$ instead, then $\ell$ would have been fixed by $A(t-1)$ (because $A(t-1)$ fixes all bits of $\Cl(A(t-1))$). Therefore, in the conversion from the original PDT to its canonical version, this query would not have been made.} \newline

By Lemma \ref{uniform_coset}, we have $\text{Pr}_{u \in G^{-1}(z_2)} [u \in A(t-1)_{\beta(t-1)}] \in [1 \pm n^{-1-\eta/4}] 2^{-k}$ and $\text{Pr}_{u \in G^{-1}(z_2)} [u \in (A')_{\beta(t-1)}] \in [1 \pm n^{-1-\eta/4}] 2^{-k-1}$.
Thus, we have

\begin{align} \text{Pr}_{x \leftarrow G^{-1}(z)} [x \in A' | x \in A(t-1)] & = \dfrac{\text{Pr}_{u \in G^{-1}(z_2)} [u \in (A')_{\beta(t-1)}]}{\text{Pr}_{u \in G^{-1}(z_2)} [u \in A(t-1)_{\beta(t-1)}]} \\ 
& \in \left[ \dfrac{1}{2} - n^{-1-\eta/6}, \dfrac{1}{2} + n^{-1-\eta/6} \right], \end{align}
so the statistical distance between the two distributions is at most $O(n^{-1-\eta/6})$. \newline

\end{mdframed}
\end{proof}

Now we prove the main result of this section: if $\Phi$ is $(p,p+q)$-DT-hard then $(\Phi, g)$ is $(p,q)$-PDT-hard whenever $g: \F^b \rightarrow \F$ is a gadget with $||\hat{g}||_{\infty} \leq n^{-2-\eta}$ (for some constant $\eta > 0$) and $b= O(\log n)$.

\begin{proof} (of Theorem \ref{thm: lifting})
    Let $P \subseteq \{0,1,*\}^n$ be a set of partial assignments and let $\mu_{\rho}$ be a DT-hard distribution for each $\rho \in P$ with $|\rho| \leq p$. This same family of distributions will establish PDT-hardness. \newline

    Assume this family is not PDT-hard. Then, the following objects exist: \newline

   \begin{mdframed}
       
    \begin{itemize}
        \item An affine space $A \subseteq \F^{nb}$ with $\Clh(A) \leq p$
        \item A closure assignment $y \in \F^{\Cl(A) \times [b]}$ such that $\rho= G(y) \in P$. Let $\mu_{\rho}$ denote the hard distribution for $\rho$ guaranteed to exist by DT-hardness.
        \item A PDT $T$ with $\text{depth}(T) \leq q$ such that, as $x \leftarrow G^{-1}(\mu_{\rho}) \cap C_y \cap A$, with probability $\geq 2/3$, the following holds:
        \begin{itemize}
            \item Let $A(q)$ be the affine space seen by $T$ after $q$ queries, i.e., \newline
            $$A(q)=C_y \cap A \cap \{ \text{intersection of all queries made}\}$$ \newline
            Then, $G(x)|_{\Cl(A(q)} \not \in P$.
        \end{itemize}
    \end{itemize}

       \end{mdframed}
\vspace{2mm}

    WLOG assume $T$ only queries linear forms supported on blocks in $[n] \setminus \Cl(A)$ (since the closure assignment $y$ is fixed). Note that $A_y$ is a safe affine subspace. Our goal is to contradict the original assumption of $(p,p+q)$-DT-hardness - so we shall construct an ordinary decision tree $T_{\text{ord}}$ such that starting from the cube $C_\rho$, when $z$ is sampled from $\mu_{\rho}$, within $p+q$ queries $T_{\text{ord}}$ sees a partial assignment not in $P$. \newline
    
   Let $\mathcal{T}$ be the randomized decision tree which simulates $\text{CANONIZE}(T, A_y)$ by making at most $p+q$ queries, guaranteed to exist by Theorem \ref{thm: simulation}. Let $z= (z_1, z_2)$ where $z_1 \in \F^{\Cl(A)}, z_2 \in \F^{[n] \setminus \Cl(A)}$. By Theorem \ref{thm: simulation}, the output distribution of $\mathcal{T}(z_2)$ is $o(1)$-close to the transcript of $\text{CANONIZE}(T, A_y)$ when run on a uniformly random preimage of $z_2$ conditioned with $A_y$. Let $\beta$ be the closure assignment at the last step of the transcript. When $z$ is sampled according to $\mu$, with probability $\geq 2/3 - o(1)$ (over the randomness of $z$ and $\mathcal{T}$), the partial assignment $G(\beta, y)$ does not lie in $P$. Consider the randomized ordinary decision tree $T_{\text{rand}}$ which simply runs $\mathcal{T}$ and outputs $G(\beta, y)$. With probability $\geq 1/2$, $T_{\text{rand}}$ sees a partial assignment not in $P$ after $p+q$ queries. By fixing the coins of $T_{\text{rand}},$ we get a deterministic ordinary decision tree $T_{\text{det}}$ which sees a partial assignment not in $P$ with probability $\geq 1/2$. This contradicts the assumption of $(p,p+q)$ hardness.

\end{proof}

\subsection{DT-hardness implies lower bounds for depth-restricted $\text{Res}(\oplus)$}
Combining Theorem \ref{random_walk_with_restart_official} and Theorem \ref{thm: lifting}, and replacing $\eta$ by $\eta/2$ (for later convenience), we obtain the following:

\begin{theorem} \label{ref: DT_hardness_implies_lower_bounds}
Let $\Phi$ be a $(p, p+q)$-DT-hard CNF. Let $g: \F^b \rightarrow \F$ be a gadget with $b= O(\log n)$ and $||\hat{g}||_{\infty} \leq n^{-2-\eta/2}$ for some constant $\eta>0$. Then, any $\text{Res}(\oplus)$ refutation of $\Phi \circ g$ of size $s$ must have depth $\Omega \left(  \dfrac{pq}{\log (s)}\right)$.
    
\end{theorem}

%% file: proving_dt_hardness.tex
The main theorem of this subsection is that Tseitin contradiction over an expander is $(\Omega(n), \Omega(n))$-DT hard, where $n$ is the number of variables. \newline

For convenience of the reader, we recall the definition of DT-hardness here. \newline

\textbf{Definition \ref{def: unlifted_hardness}}: 
     Let $\Phi$ be a CNF on $n$ variables, call a set of partial assignments $P \subseteq \{0,1,*\}^n$ to be $(p,q)$-DT-hard wrt $\Phi$ if the following hold:
    \begin{itemize}
    \item \textbf{Non-emptiness: }$P \neq \phi$
    \item \textbf{No falsification:} No partial assignment $\rho \in P$ falsifies any clause of $\Phi$.
    \item \textbf{Downward closure: } For any $\rho \in P$ and any $j \in [n]$, if $\tilde{\rho}$ is obtained by setting $\rho(j) \leftarrow *,$ then $\tilde{\rho} \in P$
    \item \textbf{Hard for decision trees: }For any $\rho \in P$ which fixes at most $p$ variables, there exists a distribution $\mu_{\rho}$ on assignments to unfixed variables  such that the following holds:
    \begin{itemize}
        \item Let $T$ be a decision tree of depth $q$ querying the unfixed variables. If we sample an assignment to the unfixed variables from $\mu_{\rho}$ and run $T$ for $q$ steps, the partial assignment seen by the tree also lies in $P$ with probability $\geq 1/2$.
    \end{itemize}
    \end{itemize}

    The CNF $\Phi$ is $(p,q)$-DT hard if it admits a set of $(p,q)$-DT-hard partial assignments.

\begin{theorem} \label{thm: DT_hard}
    Let $\Phi$ be the Tseitin contradiction over an $(m, d, \lambda < 0.95) $ expander (with $m$ odd). Then, $\Phi$ is $(m/2000, m/2000)$-DT-hard -- i.e., there exists a non-empty $(m/2000, m/2000)$-DT-hard set of partial assignments for $\Phi$.
\end{theorem}

\begin{remark}
We assume throughout that $d \ge 3$ as otherwise no expander can be constructed.
\end{remark}
\subsection{Choosing the set of partial assignments} \label{subsubsec: hard_assignments}

We define a partial assignment to the edges of our graph to be \emph{valid} when it satisfies conditions given below. The set $P \subseteq \{0,1,*\}^n$ will be the set of valid partial assignments. Recall, $\text{fix}(\rho)\triangleq \{e | e \text{ has been fixed by }\rho \}$. For each $v \in V$, define $f_{\rho}(v) \triangleq 1+ \displaystyle \sum_{(v,w) \in \text{fix}(\rho)} \rho(v,w)$ (i.e., $f_{\rho}(v)$ denotes the parity of the unfixed edges incident to $v$ in order to satisfy the original degree constraint for $v$).
\newline

\begin{definition}
\label{valid_assignments}

Let $\rho \in \{ 0,1,*\}^{E}$ be a partial assignment. We say $\rho$ is \emph{valid} if it satisfies the following:   
\begin{enumerate}
    \item There exists exactly one connected component $C$ in $(V, E \setminus \text{fix}(\rho))$ such that $\displaystyle \sum_{v \in C} f_{\rho}(v) \equiv 1 \pmod{2}$. We call this component the \emph{odd} component and every other component is called \emph{even}. 
    \item The size of the odd connected component, $|C|$, is more than $m/2$.
    
\end{enumerate}
\begin{mdframed}
\begin{center}\input{big_comp}\end{center} \vspace{1mm}
Example of a $\rho \in P$: in $(V, E \setminus \text{fix}(\rho))$, exactly one component has $\sum f$ odd -- and the size of that component is large
\end{mdframed}
\end{definition}
We now show that the set $P$ of all valid partial assignments is indeed $(p,q)$-DT-hard.
We begin by showing below that the first two properties for being $(p,q)$-DT-hard are satisfied.

\begin{lemma} \label{lem: first_two_cond_satisfied}
The set of partial assignments $P$ satisfies the conditions \textbf{Downward Closure} and \textbf{No falsification} for $\Phi$ (as defined in Definition \ref{def: unlifted_hardness}).    
\end{lemma}
\begin{proof} Both properties are straightforward to verify.

\begin{itemize}
    \item \textbf{No falsification: }In order to falsify any clause, $\rho$ has to fix all edges of some vertex. In that case, that vertex is an isolated connected component in $(V, E \setminus \text{fix}(\rho))$ and the total $f_{\rho}$ in that component is $1 \pmod{2}$. However, the first condition stipulates that there is exactly one connected component whose total $f_{\rho}$ is odd, and that component has size  more than $m/2$.
    \item \textbf{Downward closure: } Let $\rho \in P$, and let $\tilde{\rho}$ be obtained from $\rho$ by setting $\rho(e) = {*}$ for some $e \in E$ that was fixed by $\rho$. There are three cases.
    \begin{enumerate}
        \item $e=(a,b)$ bridges the largest component $C$ with some other component $W$. Wlog, $a \in C$ and $b \in W$. We have $\text{fix}(\tilde{\rho}) = \text{fix}(\rho) \setminus \{e\}$. 
        Let the new expanded connected component be $C'= C \cup W$. Note that $W$ forms a connected component in $(V, E \setminus \text{fix}(\rho))$ and therefore $\displaystyle \sum_{v \in W} f_{\rho}(v) = 0 \pmod{2}$. 
        For each $v \neq a,b$, $f_{\tilde{\rho}}(v)=f_{\rho}(v)$, and $f_{\tilde{\rho}}(a)=\rho(e)+f_{\rho}(a) \pmod{2}, f_{\tilde{\rho}}(b)= \rho(e)+f_{\rho}(b) \pmod{2}$. We have $|C'| \geq |C| \geq m/2$, so all we need to verify is that $\displaystyle \sum_{v \in C'} f_{\tilde{\rho}}(v) \equiv 1 \pmod{2}$. 
        
        $$ \displaystyle \sum_{v \in C \cup W} f_{\tilde{\rho}}(v) = \displaystyle \sum_{v \in C } f_{\rho}(v) + \rho(e) + \displaystyle \sum_{v \in W} f_{\rho}(v) + \rho(e) = 1 + 0 = 1 \pmod{2}$$
        \item $e = (a,b)$ bridges two components $U$ and $W$, none of which is $C$. Very similar argument as above shows that  
        $$ \displaystyle \sum_{v \in U \cup W} f_{\tilde{\rho}}(v) = \displaystyle \sum_{v \in U } f_{\rho}(v) + \rho(e) + \displaystyle \sum_{v \in W} f_{\rho}(v) + \rho(e) = 0 + 0 = 0 \pmod{2}$$
        Thus, $C$ remains the unique odd connected component.
        \item $e=(a,b)$ does not bridge two different components. It is simple to verify in this case that the parity of all components remain unchanged.
    \end{enumerate}
    
\end{itemize}
\end{proof}

Now we come to the final property: hardness for decision trees. \newline Before proceeding, we provide an informal proof overview. Readers interested in the actual proof can skip ahead to the end of this box.

\begin{mdframed}
    \textbf{Informal overview of proof of Theorem \ref{unlifted_tree_hard}} 
    \vspace{1mm}
    \hrule
    \vspace{1mm}
    We started with the CNF $\Phi$ being the Tseitin contradiction on an expander, with a parity constraint $\displaystyle \sum_{w \in N(v)} z(v,w) \equiv 1 \pmod{2}$ for each $v \in V$. Given a partial assignment $\rho$, some of the variables get fixed. Now for each $v \in V$ we have a parity constraint over its unfixed incident edges:
   $$  \displaystyle \sum_{\substack{w \in N(v) \\ (w,v) \not \in \text{fix}(\rho)}} z(v,w) \equiv f_{\rho}(v) \pmod{2}$$
   This is another Tseitin contradiction on the graph where the edges in $\text{fix}(\rho)$ are deleted (we call this graph $G_{\rho})$. For notational convenience, for a set of vertices $S \subseteq V$ call $f_{\rho}(S)= \displaystyle \sum_{v \in S} f_{\rho}(v)$. Let $C_1, C_2, \cdots , C_k$ be the connected components of $G_{\rho}$. The conditions for $\rho \in P$ are following: exactly one of $C_1, C_2, \cdots , C_k$ has $f_{\rho}(C_j) \equiv 1 \pmod{2}$, and moreover, the size of that component is at least $m/2$. \newline

   Suppose $C_1$ is the unique odd component. There is no assignment which simultaneously satisfies all parity constraints of $C_1$ - so $C_1$ must contain a falsified clause. We have to design a hard distribution $\mu_{\rho}$, Informally speaking, we must make it hard for a prover to locate a falsified clause. \newline
   
   It can be shown (Corollary \ref{all_but_one}) that there exists an assignment $z_{\text{sat}}$ which simultaneously satisfies the parity constraint of every vertex that is not in $C_1$. Since our goal is to make locating a falsified clause harder, it makes sense to not have any falsified clause in $C_2, C_3, \cdots , C_k$ in the first place: we assign those edges according to $z_{\text{sat}}$. Now comes the interesting part: how do we define a distribution on the values of the edges in $C_1$? \newline

   Our procedure for this is simple: we pick a uniformly random vertex $v$. We shall call this vertex the \textit{root}. One can show that there exists an assignment which satisfies all parity constraints other than $v$ (Corollary \ref{all_but_one}). After picking $v$, we uniformly at random pick one such assignment. As we shall see soon, a decision tree that tries to home in on a partial assignment not in $P$ is effectively trying to locate the root.
   \newline

   For simplicity let us assume the decision tree does not query any edges in $C_2, C_3, \cdots, C_k$ (there are no falsified clauses there, so the tree would just be wasting its budget). When it queries an edge in $C_1$, we delete that edge from the graph. Now, when the tree queries an edge $e$, the component $C_1$ might get split into two components: $C_{new}^{(0)}$ and $C_{new}^{(1)}$. Note that after getting the response to the query, $f$ also gets updated (Recall that $f(v)$ is the RHS of the parity constraint at $v$. When a variable gets queried and determined, it goes to the RHS.) Out of $C_{new}^{(0)}$ and $C_{new}^{(1)}$, the component that contains the root will have total $f$ odd, and the other one will have total $f$ even. Reason: for the component that does not contain the root, there is an assignment which satisfies all its parity constraints, so its total $f$ is even. The total $f$ of the other component cannot be even, as otherwise there would be an assignment that satisfies the parity constraint at each vertex. \newline

   So at any point of execution of the decision tree, the graph (with all the queried edges deleted) will have exactly one component with $f$ odd - and that will be the component containing the root. If the tree wants to see a partial assignment that does not lie in $P$, it needs to ensure the unique odd component has size $\leq m/2$. In other words, it needs to narrow down the possibility of the root to at most $m/2$ vertices. \newline

   Given this, our decision to pick the root uniformly at random looks wise - we want there to be as much uncertainty about its location as possible. But what about the situation when the tree has queried some edges, and therefore has some idea about the location of the root? We need to analyze exactly what information the tree has about the location of the root. For one, the tree knows that the root must lie in the current odd component. We shall show that this is all that the tree knows about the location of the root: when $z$ is sampled from $\mu$ conditioned on $[\mathrm{current \ information \ obtained \ by \ the \ tree}]$, the conditional distribution of the root is uniform on the odd component (Lemma \ref{root_is_hidden_unlifted}). \newline
   
   Now essentially the game looks as follows: at any point of time, the tree has seen some partial assignment - and it knows that the root lies in the unique odd component. As the game proceeds and more edges are queried (and then deleted), the unique odd component keeps shrinking. The goal of the decision tree is to ensure the size of the odd component goes below $m/2$. \newline

    To complete the proof, we need one last ingredient about expander graphs. Using the isoperimetric profile of the graph, one can show that if $\leq m/1000$ edges are deleted from $G$, one of its components must have large size: $\geq m \left( 1 - \dfrac{1}{2d} \right)$. So if the decision tree is to put the root in a small component, it must actually put the root in a \textit{very small} component. \newline

    Now consider a point of time when the large component shrinks. Before the bridge is queried, the conditional distribution of the root is uniform on this entire component. \newline
    \input{bridge_deleted_edited} \newline
    
    Therefore, after the bridge is queried and it is revealed in which component the root lies, with overwhelmingly high probability the root is going to end up in the larger component. With high probability, this is going to happen every time. To formalize this intuition, we do an amortized induction where we count the amount by which the largest component has shrunk instead of total number of queries made.

\end{mdframed}

First, we prove an easy but crucial lemma. This lemma is standard and has appeared in the literature before (for example in \cite{U87}). \newline

\begin{lemma}
    \label{spanning_tree}
    Let $G= (V,E)$ be a connected undirected graph and let $f: V \rightarrow \F$ a parity constraint for each vertex. Let $T \subseteq E$ be a spanning tree, and let $\tilde{z} \in \F^{E \setminus T}$ be an assignment to the edges not in $T$. Let $v \in V$ be a vertex. There exists a unique assignment $z\in \F^{E}$ such that $z$ extends $\tilde{z}$ and $\displaystyle \sum_{w \in N(u)} z(u,w) = f(u)$ for all $u \neq v$. 
\end{lemma}

\begin{proof}
    We construct $z$ as follows: convert $T$ to a rooted tree by making $v$ the root and then process the vertices bottom up, starting at the leaves of $T$. When vertex $u$ is being processed, all edges in the subtree of $u$ have been assigned. Then, exactly one edge incident to $u$ is kept unfixed (the edge $(u,\text{parent}[u])$ - assign it so that $\displaystyle \sum_{(u,w) \in E} z(u,w)= f(u)$ is satisfied. \newline

    It is also clear that this is the unique assignment to the edges in $T$ which satisfies all these constraints and is consistent with the assignment to the edges of $E\setminus T$. This is because once the edges in the subtree of $u$ has been fixed, there is a unique choice of $z(u, \text{parent}[u])$ that satisfies the parity constraint of $u$.\newline
\end{proof}

We call the procedure in the proof of Lemma \ref{spanning_tree} $\text{FIX}(G, T, v, \tilde{z}, f)$.

\begin{mdframed}
\vspace{1mm}
$\text{FIX}(G, T, v, \tilde{z}, f)$ \vspace{1mm}
\hrule \vspace{1mm}
\textbf{Input:}
\begin{itemize}
    \item $G= (V,E)$: a connected graph
    \item $T$: a spanning tree of $G$
    \item $v$: a vertex in $G$
    \item $\tilde{z} \in \F^{E \setminus T}$: an assignment to the edges not in $T$
    \item $f: V \rightarrow \F$: a parity constraint for each vertex
    
\end{itemize}
\hrule
\vspace{1mm}
\textbf{Output}:
\begin{itemize}
    \item The unique assignment $z \in \F^{E}$ such that $z$ extends $\tilde{z}$ and $\displaystyle \sum_{w \in N(u)} z(u,w) = f(u)$ for all $u \neq v$
\end{itemize}
\hrule
\vspace{1mm}
\textbf{Algorithm: }
\begin{itemize}
    \item Set $z(e)= \tilde{z}(e) \forall e \not \in T$
    \item Root $T$ at $v$. Let $\mathrm{children}(u)$ be the set of children of $v$.
    \item Execute $\text{RECURSIVE-FIX}(v)$
    \item Return $z$
\end{itemize}

Here $\text{RECURSIVE-FIX}(u)$ is implemented as follows (with $T, G, z, f$ being global variables and the vertex $u$ being passed as a parameter): \newline

$\text{RECURSIVE-FIX}(u)$:
\begin{itemize}
    \item For $w \in \mathrm{children}(u)$:
    \begin{itemize}
        \item Execute $\text{RECURSIVE-FIX}(w)$.
        \item Set $z(u,w)= f(w) + \displaystyle \sum_{u' \in N(w) \setminus \{u \}} z(u', w)$
    \end{itemize}
\end{itemize}

\end{mdframed}

We make the following observation now.

\begin{observation} 
\label{obs:fix_works_good} Let $G$ be a connected graph, $T$ a spanning tree of $G$, $v \in V$ any vertex, $\tilde{z} \in \F^{E(G) \setminus E(T)}$ an assignment to the non-tree edges, and $f:\F^V \rightarrow \F$ a parity constraint for each vertex.
\begin{itemize}
    \item If $\displaystyle \sum_{u \in V} f(u) = 0 \pmod{2}$ and $z= \text{FIX}(G, T, v, \tilde{z}, f)$, then $\displaystyle \sum_{w \in N(v)} z(v,w) = f(v) \pmod{2}$
    \item If $\displaystyle \sum_{u \in V} f(u) = 1 \pmod{2}$ and $z= \text{FIX}(G, T, v, \tilde{z}, f),$ then $\displaystyle \sum_{w \in N(v)} z(v,w)\neq f(v) \pmod{2}$
\end{itemize}

\end{observation}

\begin{corollary}
\label{all_but_one}
    Let $G$ be a connected undirected graph and let $f: V \rightarrow \F$ be any map.
    \begin{enumerate}
        \item If $\displaystyle \sum_{v \in V} f(v) = 1 \pmod{2}$, then for any $v \in V$ there exists an assignment $z \in \F^E$ such that $\displaystyle \sum_{w \in N(u)} z(u,w) = f(u) \pmod{2}$ for each $u \neq v$, and $\displaystyle \sum_{w \in N(v)} z(v,w) \neq f(v) \pmod{2}$.
        \item If $\displaystyle \sum_{v \in V} f(v)= 0 \pmod{2}$, there exists an assignment $z \in \F^E$ such that $\displaystyle \sum_{w \in N(u)} z(u,w) = f(u) \pmod{2}$ is satisfied for all $u \in V$.
    \end{enumerate}
    
\end{corollary}

Now we define the following hard distribution for each $\rho \in P$, when $|\rho| \leq \dfrac{m}{2000}$.\newline

\subsubsection{The hard distrbution}
Our goal in this subsection is to define for each $\rho \in P$ a distribution $\mu = \mu_{\rho}$ on the unfixed variables so that the requirement in Definition \ref{def: unlifted_hardness} for DT-Hardness is satisfied.
\begin{definition}
\label{hard_distribution_unlifted}  
Let $\rho \in P$ be a valid partial assignment which fixes at most $\dfrac{m}{1000}$ edges. Define the following:

\begin{enumerate}
    \item Define $f_{\rho}$ as before (i.e. for each $v\in V$, $f_{\rho}(v) = 1 + \displaystyle \sum_{(u,v) \in \text{fix}(\rho)} \rho(u,v)$. After the edges in $\text{fix}(\rho)$ are fixed according to $\rho$, $f_{\rho}(v)$ is the RHS of the modified parity constraint of $v$. 
    \item Let $C_{\rho}$ be the unique connected component in $G_{\rho}=(V, E \setminus \text{fix}(\rho))$ whose total $f_{\rho}$ is odd.

\end{enumerate}
\end{definition}

\textbf{Notation:} For a subset $A \subseteq V$ we denote by $E(A)$ the set of edges in $G$ both of whose endpoints lie in $A$. \newline

Now we describe the procedure of sampling from $\mu$. Note that the values of edges in $\text{fix}(\rho)$ are fixed; we have to define a distribution on variables in $\text{free}(\rho)$. We do this as follows. 
\newline


\fbox{
\begin{minipage}{40em}
\hspace{1mm}

\textbf{DTFooling}   \hspace{1mm}

\textbf{Input: }\begin{itemize}
    \item Graph $G=(V,E)$
    \item A valid partial assignment $\rho \in \{0,1,*\}^{E},$ $\rho \in P$
\end{itemize}

 \hspace{1mm}

\textbf{Output: }A sample $z \in \{0,1\}^E$ from $\mu_{\rho}$ \newline 
 \hspace{1mm}

\textbf{Sampling procedure}
\begin{itemize}
    \item \textbf{Assigning the edges of $C_{\rho}$}: Fix an arbitrary spanning tree $T$ of $C_{\rho}$ formed by edges from $\text{free}(\rho)$. Uniformly at random pick a vertex $v \in C_{\rho}$. Let $\tilde{w} \in \F^{E(C_{\rho}) \setminus (\text{fix}(\rho) \cup E(T))}$ be a uniformly random assignment to the unfixed edges in $C_{\rho}$ not in $T$. Let $\tilde{z} \in \F^{E(C_{\rho}) \setminus E(T)}$ be the following assignment to the non-tree edges in $C_{\rho}$:
    \begin{align*}
        \tilde{z}(e)= \begin{cases}
            \rho (e) & \text{ if }e \in E(C_{\rho}) \cap \text{fix}(\rho) \\
            \tilde{w}(e) & \text{ if }e \in (E(C_{\rho}) \cap \text{free}(\rho)) \setminus E(T)
        \end{cases}
    \end{align*}
    
    Assign the edges in $C_{\rho}$ according to $\text{FIX}(C_{\rho}, T, v, \tilde{z}, f_{\rho}).$ 
    
    \item \textbf{Assigning the edges for every other component $C'$:} Pick an arbitrary spanning tree $T'$ of $C'$ formed by edges from $\text{free}(\rho)$. Let$\tilde{w}' \in \F^{E(C') \setminus \{\text{fix}(\rho) \cup E(T') \}}  $ be a uniformly random assignment to the free non-tree edges. Let $\tilde{z}' \in \F^{E(C') \setminus E(T')}$ be the following assignment to the non-tree edges in $C'$:
    \begin{align*}
        \tilde{z}'(e) = \begin{cases}
            \rho(e) & \text{ if }e \in E(C') \cap \text{fix}(\rho) \\
            \tilde{w}'(e) & \text{ if }e \in (E(C') \cap \text{free}(\rho)) \setminus E(T')
        \end{cases}
    \end{align*}
    
    Pick an arbitrary vertex $v'$. Assign the edges of $C'$ according to $\text{FIX}(C', T', v', \tilde{z}', f_{\rho})$. 
\end{itemize}
\end{minipage} 
}

Let $A_{\rho,v}$ denote the set of all $z \in \mathbb{F}_2^{E}$ that are consistent with $\rho$ and satisfy the following:
\begin{enumerate}
        \item For all $u \neq v, \displaystyle \sum_{w \in N(u)} z(u,w)= f_{\rho}(u)$.
        \item $\displaystyle \sum_{w \in N(v)} z(v,w) = 1 + f_{\rho}(v)$.
\end{enumerate}
We make the following remark now.
\begin{remark}
    \label{uniform_sampling}
    Let $\rho$ be any valid (partial) assignment to edges of $G$. Then,
    \begin{enumerate}
        \item  $A_{\rho,v}$ is an affine space in $\mathbb{F}_2^E$, for each $v \in C_{\rho}$.
        \item \textit{DTFooling} picks a random $v \in C_{\rho}$ and then samples a random point in $A_{\rho,v}$.
    \end{enumerate}
    
\end{remark}

For any $z \in \text{supp}(\mu)$, the parity constraint is violated for exactly one vertex (the vertex which was chosen as the root of the spanning tree of the odd connected component). Call this vertex $\text{\emph{root}}(z)$.

Before proving the hardness, we note down some properties of the distribution. \newline

\subsubsection{Conditional Distribution of the Root is Uniform}
We prove a useful property of the distribution sampled by \DTFooler, given a valid partial assignment $\rho$. The idea is when a decision tree queries bits from an assignment $z$ to the edges sampled according to $\mu_\rho$, the graph $G_\rho$ starts splitting into further smaller components. The decision tree knows at every instant in which component $\text{root}(z)$ lies,as there is always a unique odd component.  The lemma below ensures that conditioned on what the decision tree has observed so far, the distribution of  $\text{root}(z)$ remains uniform over all vertices in the odd component. \newline

In the following let $\text{free}(\rho) \subseteq E$ be the set of edges free in $\rho$. For convenience of the reader, we re-state the definition of  $f_{\rho}$ in Definition \ref{valid_assignments} here.
\begin{itemize}
    \item $f_{\rho} \in \F^V, f_{\rho}(v)= 1 + \displaystyle \sum_{(v,w) \in\text{fix}(\rho)} \rho(v,w)$
\end{itemize}

\begin{lemma}
\label{root_is_hidden_unlifted}    
Let $\rho$ be a valid partial assignment and $\mu = \mu_{\rho}$ be the distribution in Definition \ref{hard_distribution_unlifted}. Let $S \subseteq \text{free}(\rho)$ be a subset of free edges and let $\alpha \in \F^{S}$ be an assignment to $S$. Define $f_{\alpha}(v)= f_{\rho}(v) + \displaystyle \sum_{(v,w) \in S} \alpha(v,w)$. Suppose the connected components of $(V, \text{free}(\rho) \setminus S)$ are $C_1, C_2, \cdots , C_k$, where $\displaystyle \sum_{v \in C_1} f_{\alpha}(v)=1$, and for all $j \neq 1,$ $\displaystyle \sum_{v \in C_j} f_{\alpha}(v)=0$. Then, the following is true:
\begin{itemize}
\item The distribution of $\text{root}(z)$ as $z$ is sampled from $\mu= \mu_{\rho}$ conditioned on $z|_{S}= \alpha$, is uniform on $C_1$. \footnote{In order for the statement to be valid, we need that some assignment in $\text{supp}(\mu)$ is consistent with $\alpha$ (otherwise we are conditioning on a probability 0 event). This follows from Corollary \ref{all_but_one} applied to each component of $(G, \text{free}(\rho) \setminus S)$ separately.}
\end{itemize}
\end{lemma}

This lemma is essentially saying the following: suppose the input is sampled from $\mu$ and currently, the tree has seen a partial assignment. After fixing the edges seen by the tree, the graph splits into multiple components and there is a unique odd component $C_1$. Then, when $z$ is sampled from $\mu \text{ conditioned on being consistent with current information obtained}$, the distribution of $\text{root}(z)$ is uniform on $C_1$.

\begin{proof}
    Conditioned on $z|_S = \alpha$, the root cannot lie in any of $C_2, C_3, \cdots , C_k$. (Reason: after we choose the root, only the parity constraint at the root is violated; other parity constraints are satisfied. But after fixing $S$ to $\alpha$, it is not possible to satisfy all parity constraints of $C_1$ simultaneously as the sum of the modified parity constraints of $C_1$ is odd.) \newline

    So we need to show that for all $u \in C_1,$ $\text{Pr}_{\mu} [\text{root}(z)=\ u | z_S = \alpha]$ is a non-zero quantity independent of $u$. Note that $C_1 \subseteq C_{\rho}$ (recall that $C_{\rho}$ is the unique odd component of $G_{\rho}$). Since for all $u \in C_1$, $\text{Pr}_{\mu} [\text{root}(z)=u] =\dfrac{1}{|C_{\rho}|}$ is a non-zero quantity independent of $u$, by Bayes' rule it suffices to show that for all $u \in C_1$, $\text{Pr}_{\mu} [z_{S} = \alpha | \text{root}(z)=u]$ is a non-zero quantity independent of $u$. \newline

    Let $M \in \F^{V \times E}$ be the edge-vertex incidence matrix of $G$. Let $\gamma_v \in \F^{V}$ be the following vector:
    $$ \gamma_v(u)= \begin{cases}
    f_{\rho}(u) & \text{ if }{ u \neq v } \\
    1+f_{\rho}(u) & \text{ otherwise}
        \end{cases}
     $$
     Once $v$ is chosen as the root, the sampling procedure samples a uniformly random element of the affine space $\{z | Mz= \gamma_v\}$. \newline

     Let $S= \{r_1, r_2, \cdots , r_{|S|}\}$. Let $N \in F^{S \times E}$ be the matrix whose $j$-th row is the standard basis vector at coordinate $r_j$. Once $u$ is chosen as the root, $z|_S=\alpha$ if and only if $z$ satisfies the following equation:
   $$ \begin{bmatrix}
    M \\
    \hline
    N 
\end{bmatrix} z = \begin{bmatrix} \gamma_u \\ \hline \alpha\end{bmatrix}$$
Let
$$J = \begin{bmatrix} M \\ \hline N\end{bmatrix}
$$
Conditioned on satisfying $Mz=\gamma_u$, the probability of satisfying $z|_S=\alpha$ is either $2^{\text{rank}(M) - \text{rank}(J)}$ (if there is a solution) or 0 (if there is some inconsistency in the right hand sides of the system of equations). (Indeed, for $v \not \in C_1$ there is an inconsistency in the right hand sides - as noted above.) \newline

Now for all $v \in C_1$, we shall show there is a $z$ such that $\text{root}(z)= v$ and $z|_{S}= \alpha$ -- this will show that for $v \in C_1$, $\text{Pr} [z|_S= \alpha | \text{root}(z)=v]$ is a non-zero quantity independent of $v$. To show this, we construct an assignment as follows:
\begin{itemize}
    \item For each $C_j$, choose a spanning tree $Q_j$ disjoint from $S \cup \text{free}(\rho)$. 
    \item Let $\tilde{z}$ be any assignment to the edges not in $Q_1, Q_2, \cdots, Q_k$ such that $\tilde{z}$ agrees with $\alpha$ on $S$.
     \item Assign the values of the edges in $C_1$ according to $\text{FIX}(C_1, Q_1, v, \tilde{z}, f_{\alpha})$
    \item For $j>1$, pick an arbitrary vertex $v_j \in C_j$ and assign the edges of $Q_j$ according to $\text{FIX}(C_j, T_j, v_j, \tilde{z}, f_{\alpha})$
\end{itemize}
This assignment satisfies $z|_{S}= \alpha$ by construction; moreover, it also satisfies $ Mz= \gamma_v$ because of Lemma \ref{spanning_tree} and Observation \ref{obs:fix_works_good}. This completes the proof.
\end{proof}

\subsubsection{Proving Hardness}

Now we prove $(p,q)-$hardness for ordinary decision trees. \newline

\textbf{Convention: }Henceforth, given a partial assignment $\rho$ and a decision tree $T$ making queries to the variables in $\text{free}(\rho)$, we define \textit{the partial assignment seen by the tree} to be the partial assignment formed by fixing the edges queried by the tree and $\rho$. \newline

\begin{theorem}
\label{unlifted_tree_hard}
Let $G= (V,E)$ be an $(m, d, \lambda < 0.95)$-spectral expander with $|V|=m$ being odd. Let $P$ be the set of partial assignments defined as in Definition \ref{valid_assignments}. Let $\rho \in P$ be a valid partial assignment with $|\rho|=p \leq \dfrac{m}{2000}$. Let $\mu$ be the distribution defined as in Definition \ref{hard_distribution_unlifted} \footnote{$\mu$ is a distribution on $\F^n$ such that every $z \in \text{supp}(\mu)$ is consistent with $\rho$}. Let $T$ be any decision tree making at most $q \leq \dfrac{m}{2000}$ queries. Sample $z \leftarrow \mu$. Then, with probability $\geq 1/2$, the partial assignment seen by the tree after $q = m/2000$ queries also lies in $P$.
\end{theorem}

\textit{Notice the slightly unconventional notation here: $m$ denotes the number of vertices in the graph, and $n$ denotes the number of edges (i.e., the number of variables in the Tseitin contradiction), ($n=md/2= \Theta(n)$)}. \newline

\begin{proof}
    We fix some notation that will be used in the rest of the proof.

    \begin{enumerate}
        \item At time-step $j$, the partial assignment seen by the tree is $\rho_j$ (this includes the edges fixed by $\rho$ and the edges queried by $T$).
        \item $E_j \triangleq \text{fix}(\rho_j)$
        \item $G_j \triangleq (V, E \setminus E_j)$
        \item Define $f_j(v) = 1 + \displaystyle \sum_{(v,w) \in E_j} \rho_j(v,w)$. 
        \item Let $C_j$ be the unique connected component of $G_j$ such that $\displaystyle \sum_{v \in C_j} f_j(v) \equiv 1 \pmod{2}$
        
    \end{enumerate}
    Before proceeding, we make some remarks.
    \begin{remark}
    \label{odd_comp_is_unique}
     After some vertex $v$ is chosen to be the root, it is guaranteed that the parity constraint of all vertices other than $v$ is satisfied. It is also known that all parity constraints are not satisfiable simultaneously (since sum of the right hand sides is odd). So, after we know the value of some edges (say given by the partial assignment $\sigma$), after removing those edges, exactly one connected component has odd $\sum f_{\sigma}$ - and the root lies in this component. This means that after $z$ is sampled according to $\mu$, condition (i) defining membership in $P$ is always satisfied. Only condition (ii) (which stipulates that the odd component must have large size) can possibly be violated \newline
     \end{remark}
     \begin{remark}
      Suppose after querying an edge, in $G_{j+1}$ the component $C_j$ splits into $C_j=A \cup B$. Initially $\displaystyle \sum_{v \in C_j} f_j(v)$ is odd. After querying the edge, exactly one of $\displaystyle \sum_{v \in A} f_{j+1}(v)$ and $\displaystyle \sum_{v \in B} f_{j+1}(v)$ is odd - and the root must lie in the component where the sum is odd.  \newline

      Once the value of the edge is revealed, the decision tree knows which one of $A,B$ contains the root. Thus, the decision tree has made some progress in determining the location of the root. \newline
        
       We want to say that the decision tree can never make too much progress  - our tool here is Lemma \ref{root_is_hidden_unlifted}, which says that the decision tree does not know anything about the root other than the fact that its conditional distribution is uniform on the current odd component.
    \end{remark}
3
We start with a crucial lemma.

\begin{lemma}
\label{isoperimetric}
    At any time-step $j$, the largest connected component of $G_j$ must have size $\geq m \left( 1 - \dfrac{1}{2d} \right)$
\end{lemma}

\begin{proof}
    Suppose not; let time-step $j$ be a time step where all the connected components of $G_j$ have size $< m \left( 1 - \dfrac{1}{2d} \right)$. We then greedily pick a subset of the connected components whose union $T$ has size in the interval $ \left[ \dfrac{m}{4d} , m - \dfrac{m}{4d}\right]$. Cheeger's inequality (Lemma \ref{cheeger}) then implies the cut $E(T, V \setminus T)$ has at least $\dfrac{m}{200}$ edges. \newline
    
    This means the current partial assignment fixes at least $m/200$ edges. However, the current partial assignment can only fix $p+q \leq m/1000$ edges.
\end{proof}

Now, every time the decision tree queries an edge, we make it pay us some coins as follows. Suppose the current partial assignment lies in $P$; the current graph is $G_j$ and the current odd component is $C_j$, and the tree queries the edge $e$.
\begin{itemize}
    \item If removing $e$ keeps $C_j$ connected, the tree does not have to pay anything.
    \item Suppose removing $e$ splits $C_j$ into two components: $C_j = A\cup B$. The value of $e$ is revealed - and it determines in which component of $A,B$ the root belongs to. Suppose the root lies in $A$. If $|A| \leq m/2$, the decision tree does not pay anything and wins the game. Otherwise, the decision tree has to pay $|B|$ coins.
\end{itemize}

(In other words: if, at any point of time, the largest component in $G_j$ isn't the odd component, the decision tree wins the game. Otherwise, if the decision tree shrinks the size of the largest component by $s$, it must pay $s$ coins.)

By Lemma \ref{isoperimetric}, the decision tree only pays $\leq \dfrac{m}{2d}$ coins. So we start by awarding the decision tree a budget of $b= \dfrac{m}{2d}$ coins, and argue (by induction on number of coins remaining) that the decision tree loses the game with high probability. (The decision tree loses the game when it has to pay some coins but it is broke.) \newline

At this point, we allow the decision tree to make as many as queries as it wants -- as long as it maintains that the largest component has size $\geq m \left(1-\dfrac{1}{2d} \right)$ (and therefore it does not use more than $b$ coins). We prove the following statement by inducting on number of coins remaining.

\begin{lemma}
\label{limited_budget}
    Suppose the decision tree has $c$ coins remaining and has not won the game yet. Then, the probability it wins the game is $\leq \dfrac{6c}{5m} + \dfrac{1}{5}$.
\end{lemma}

\begin{proof}
    We induct on $c$. Consider the base case $c=0$: the decision tree has no coins remaining. Let the current odd component be $C$. The first time the tree splits $C$, the root must lie in the smaller component for the decision tree to win. Suppose the tree queries an edge $e$ and $C$ splits into $C= A \cup B$ where $|A| \geq m \left( 1-\dfrac{1}{2d} \right)$. Before querying $e$, the conditional distribution of the root was uniform on $C$ by Lemma \ref{root_is_hidden_unlifted}. Conditioned on the partial assignment revealed before querying $e$, the probability the root lies in $B$ is $|B|/|C|$. The probability the tree wins the game is thus 
    $$ \dfrac{|B|}{|C|} \leq \dfrac{\dfrac{m}{2d}}{m \left( 1 - \dfrac{1}{2d} \right)} \leq \dfrac{1}{5},$$
    so the base case is true. \newline
    Now we handle the inductive step. Suppose the tree has $c$ coins. Suppose the current odd component is $C_j$, with $|C_j| \geq m \left( 1 - \dfrac{1}{2d}\right)$. Suppose the decision tree queries $e$ and removing $e$ splits $C_j$ into $C_j= A \cup B$, where $A$ is the larger component $\left(|A| \geq m \left( 1-\dfrac{1}{2d} \right)\right)$. The tree wins the game at this stage if the root lies in $B$, otherwise it pays $|B|$ coins and proceeds to the next stage. Before querying $e$ the distribution of the root is uniform on $C$, so the probability it lies in $B$ is $\dfrac{|B|}{|A|+|B|} \leq \dfrac{6|B|}{5m}$. If the root does not lie in $B$, the decision tree has $c-|B|$ coins remaining, so then it can win the came with probability at most $\dfrac{6(c-|B|)}{5m} + \dfrac{1}{5}$ by the inductive hypothesis. By union bound, the probability the tree wins the game is at most 
    $$ \dfrac{6|B|}{5m} + \dfrac{6(c-|B|)}{5m} + \dfrac{1}{5}= \dfrac{6c}{5m} + \dfrac{1}{5}$$

\end{proof}

Since the decision tree starts off with $\dfrac{m}{2d}$ coins, it can win with probability at most $\dfrac{6}{10d}+\dfrac{1}{5}$. Since $d>2$, with probability $\geq \dfrac{3}{5}$, the partial assignment seen by the decision tree after $q$ queries lies in $P$.
\end{proof}

Now we have all the ingredients require to prove $(m/2000, m/2000)$-DT-hardness of $\Phi$.
\begin{proof} \textit{(of Theorem \ref{thm: DT_hard})}
Choose the set of partial assignments $P$ as defined in Section \ref{subsubsec: hard_assignments}. We have already established that this set $P$ satisfies all requirements in Definition \ref{def: unlifted_hardness} defining DT-hardness.
\begin{itemize}
    \item \textbf{No falsification} and \textbf{Downward Closure}: Established in Lemma \ref{lem: first_two_cond_satisfied}.
    \item \textbf{Hardness against decision trees: }Established in Theorem \ref{unlifted_tree_hard}.
    
\end{itemize}
    This completes the proof.
\end{proof}

%% file: big_comp.tex
\begin{tikzpicture}[scale=.5,auto=left,every node/.style={circle,fill=blue!20}]
\node [fill=red] (n1) at (5, 0) {};
\node [fill=red] (n2) at (4.75528, 1.54508) {};
\node [fill=red] (n3) at (4.04508, 2.93893) {};
\node [fill=red] (n4) at (2.93893, 4.04508) {};
\node [fill=red] (n5) at (1.54508, 4.75528) {};
\node [fill=red] (n6) at (3.06162e-16, 5) {};
\node [fill=red] (n7) at (-1.54508, 4.75528) {};
\node [fill=red] (n8) at (-2.93893, 4.04508) {};
\node [fill=red] (n9) at (-4.04508, 2.93893) {};
\node [fill=red] (n10) at (-4.75528, 1.54508) {};
\node [fill=red] (n11) at (-5, 6.12323e-16) {};
\node [fill=red] (n12) at (-4.75528, -1.54508) {};
\node [fill=red] (n13) at (-4.04508, -2.93893) {};
\node [fill=red] (n14) at (-2.93893, -4.04508) {};
\node [fill=red] (n15) at (-1.54508, -4.75528) {};
\node [fill=red] (n16) at (-9.18485e-16, -5) {};
\node [fill=red] (n17) at (1.54508, -4.75528) {};
\node [fill=red] (n18) at (2.93893, -4.04508) {};
\node [fill=red] (n19) at (4.04508, -2.93893) {};
\node [fill=red] (n20) at (4.75528, -1.54508) {};
\node (n21) at (9.5, 0) {};
\node (n22) at (8.46353, 1.42658) {};
\node (n23) at (6.78647, 0.881678) {};
\node (n24) at (6.78647, -0.881678) {};
\node (n25) at (8.46353, -1.42658) {};
\node (n26) at (9.5, -5) {};
\node (n27) at (8.46353, -3.57342) {};
\node (n28) at (6.78647, -4.11832) {};
\node (n29) at (6.78647, -5.88168) {};
\node (n30) at (8.46353, -6.42658) {};
\node (n31) at (-6.5, 0) {};
\node (n32) at (-7.53647, 1.42658) {};
\node (n33) at (-9.21353, 0.881678) {};
\node (n34) at (-9.21353, -0.881678) {};
\node (n35) at (-7.53647, -1.42658) {};
\node (n36) at (-6.5, -5) {};
\node (n37) at (-7.53647, -3.57342) {};
\node (n38) at (-9.21353, -4.11832) {};
\node (n39) at (-9.21353, -5.88168) {};
\node (n40) at (-7.53647, -6.42658) {};
\foreach \from/\to in {n1/n2,n1/n20,n2/n3,n3/n2,n3/n4,n4/n1,n4/n2,n4/n3,n4/n5,n5/n1,n5/n2,n5/n6,n6/n1,n6/n2,n6/n3,n6/n4,n6/n5,n6/n7,n7/n1,n7/n5,n7/n8,n8/n1,n8/n6,n8/n7,n8/n9,n9/n1,n9/n3,n9/n5,n9/n6,n9/n8,n9/n10,n10/n3,n10/n4,n10/n5,n10/n6,n10/n7,n10/n11,n11/n1,n11/n3,n11/n4,n11/n6,n11/n10,n11/n12,n12/n1,n12/n2,n12/n5,n12/n6,n12/n9,n12/n11,n12/n13,n13/n1,n13/n2,n13/n5,n13/n9,n13/n10,n13/n12,n13/n14,n14/n1,n14/n4,n14/n6,n14/n9,n14/n11,n14/n15,n15/n1,n15/n2,n15/n3,n15/n4,n15/n5,n15/n6,n15/n8,n15/n10,n15/n14,n15/n16,n16/n2,n16/n5,n16/n7,n16/n9,n16/n10,n16/n11,n16/n14,n16/n15,n16/n17,n17/n1,n17/n3,n17/n5,n17/n8,n17/n9,n17/n11,n17/n12,n17/n13,n17/n16,n17/n18,n18/n1,n18/n3,n18/n4,n18/n5,n18/n7,n18/n12,n18/n13,n18/n15,n18/n16,n18/n19,n19/n4,n19/n5,n19/n7,n19/n8,n19/n10,n19/n11,n19/n12,n19/n20,n20/n2,n20/n4,n20/n6,n20/n9,n20/n11,n20/n13,n20/n15,n20/n19,n21/n22,n21/n25,n22/n23,n23/n22,n23/n24,n24/n23,n24/n25,n25/n21,n25/n22,n26/n27,n26/n30,n27/n28,n28/n26,n28/n27,n28/n29,n29/n27,n29/n28,n29/n30,n30/n26,n30/n27,n30/n28,n30/n29,n31/n32,n31/n35,n32/n33,n33/n32,n33/n34,n34/n31,n34/n35,n35/n31,n35/n32,n35/n33,n36/n37,n36/n40,n37/n36,n37/n38,n38/n37,n38/n39,n39/n36,n39/n37,n39/n38,n39/n40,n40/n37,n40/n38,n40/n39}
 \draw (\from) -- (\to );
\end{tikzpicture}

%% file: bridge_deleted_edited.tex
\begin{tikzpicture}[scale=.65,auto=left,every node/.style={circle,fill=blue!20}]
\node (n1) at (5, 0) {};
\node (n2) at (4.79746, 1.40866) {};
\node (n3) at (4.20627, 2.7032) {};
\node (n4) at (3.2743, 3.77875) {};
\node (n5) at (2.07708, 4.54816) {};
\node (n6) at (0.711574, 4.94911) {};
\node (n7) at (-0.711574, 4.94911) {};
\node (n8) at (-2.07708, 4.54816) {};
\node (n9) at (-3.2743, 3.77875) {};
\node (n10) at (-4.20627, 2.7032) {};
\node (n11) at (-4.79746, 1.40866) {};
\node (n12) at (-5, 2.83277e-15) {};
\node (n13) at (-4.79746, -1.40866) {};
\node (n14) at (-4.20627, -2.7032) {};
\node (n15) at (-3.2743, -3.77875) {};
\node (n16) at (-2.07708, -4.54816) {};
\node (n17) at (-0.711574, -4.94911) {};
\node (n18) at (0.711574, -4.94911) {};
\node (n19) at (2.07708, -4.54816) {};
\node (n20) at (3.2743, -3.77875) {};
\node (n21) at (4.20627, -2.7032) {};
\node (n22) at (4.79746, -1.40866) {};
\node (n23) at (13, 0) {};
\node (n24) at (12.4271, 1.76336) {};
\node (n25) at (10.9271, 2.85317) {};
\node (n26) at (9.07295, 2.85317) {};
\node (n27) at (7.57295, 1.76336) {};
\node (n28) at (7, 3.67394e-16) {};
\node (n29) at (7.57295, -1.76336) {};
\node (n30) at (9.07295, -2.85317) {};
\node (n31) at (10.9271, -2.85317) {};
\node (n32) at (12.4271, -1.76336) {};
\foreach \from/\to in {n1/n2,n1/n22,n2/n3,n3/n2,n3/n4,n4/n1,n4/n2,n4/n3,n4/n5,n5/n1,n5/n2,n5/n6,n6/n1,n6/n2,n6/n3,n6/n4,n6/n5,n6/n7,n7/n1,n7/n5,n7/n8,n8/n1,n8/n6,n8/n7,n8/n9,n9/n1,n9/n3,n9/n5,n9/n6,n9/n8,n9/n10,n10/n3,n10/n4,n10/n5,n10/n6,n10/n7,n10/n11,n11/n1,n11/n3,n11/n4,n11/n6,n11/n10,n11/n12,n12/n1,n12/n2,n12/n5,n12/n6,n12/n9,n12/n11,n12/n13,n13/n1,n13/n2,n13/n5,n13/n9,n13/n10,n13/n12,n13/n14,n14/n1,n14/n4,n14/n6,n14/n9,n14/n11,n14/n15,n15/n1,n15/n2,n15/n3,n15/n4,n15/n5,n15/n6,n15/n8,n15/n10,n15/n14,n15/n16,n16/n2,n16/n5,n16/n7,n16/n9,n16/n10,n16/n11,n16/n14,n16/n15,n16/n17,n17/n1,n17/n3,n17/n5,n17/n8,n17/n9,n17/n11,n17/n12,n17/n13,n17/n16,n17/n18,n18/n1,n18/n3,n18/n4,n18/n5,n18/n7,n18/n12,n18/n13,n18/n15,n18/n16,n18/n19,n19/n4,n19/n5,n19/n7,n19/n8,n19/n10,n19/n11,n19/n12,n19/n20,n20/n2,n20/n4,n20/n6,n20/n9,n20/n11,n20/n13,n20/n15,n20/n19,n20/n21,n21/n3,n21/n6,n21/n7,n21/n8,n21/n12,n21/n13,n21/n15,n21/n16,n21/n17,n21/n18,n21/n19,n21/n20,n21/n22,n22/n3,n22/n4,n22/n7,n22/n8,n22/n9,n22/n11,n22/n13,n22/n14,n22/n15,n22/n16,n22/n18,n22/n19,n22/n20,n22/n21,n23/n24,n23/n32,n24/n23,n24/n25,n25/n24,n25/n26,n26/n27,n27/n23,n27/n26,n27/n28,n28/n23,n28/n24,n28/n25,n28/n26,n28/n29,n29/n27,n29/n30,n30/n23,n30/n25,n30/n26,n30/n29,n30/n31,n31/n25,n31/n27,n31/n32,n32/n23,n32/n31}
 \draw (\from) -- (\to );

 \draw [dashed] (n1) -- (n28);
\end{tikzpicture}

%% file: putting_everything_together.tex
With Theorem \ref{ref: DT_hardness_implies_lower_bounds} and Theorem \ref{thm: DT_hard} in hand, we are now in a position to prove our main result, Theorem \ref{tseitin_lower_bound}.
 
\begin{proof} \textit{(of Theorem \ref{tseitin_lower_bound})}
Let $\Phi$ be the Tseitin contradiction on $G$. Recall that $m$ is the number of vertices in $G$ and $n$ is the number of variables in the unlifted Tseitin contradiction (i.e., $n= |E|$). Theorem \ref{thm: DT_hard} shows that that Tseitin contradiction over such an expander is $(m/2000, m/2000)$-DT-hard. Applying Theorem \ref{ref: DT_hardness_implies_lower_bounds} with $p=q=m/4000$ and denoting by $\text{IP}$ the Inner Product gadget on $(4 + \eta) \log(n)$ bits, we get the following result:
\begin{itemize}
    \item Any size $s$ refutation of $\Phi \circ \text{IP} $ must require depth $\Omega \left( \dfrac{n^2}{\log(s)}\right)$
\end{itemize}

Note that the number of variables in $\Phi \circ \text{IP}$ is $N= O(n \log(n))$. We can also interpret the result as follows:
\begin{itemize}
    \item Any depth $N^{2-\epsilon}$ $\reslin$ refutation of $\Phi \circ \text{IP}$ requires size $\exp(\tilde{\Omega}(N^{\epsilon}))$
\end{itemize}
This is what we wanted to show.
\end{proof}

%% file: appendix_A.tex
   \begin{proof} \textbf{(of Lemma \ref{closure remains same})}
$\Cl(V)=\Cl(W)$ follows from the definition of closure. \newline

We shall show $\Clh(V)= \Clh(W)$. Since amortized closure of a set of vectors depends only on its span, we can choose a different basis for $V$. Define $V_{in}= \{v \in V| \text{supp}(v) \subseteq \Cl(V)\}$. Let $B_{in}= \{a_1, a_2, \cdots , a_k\}$ be a basis for $V_{in}$. Complete this to a basis for $V$: $B= B_{in} \cup B_{out}$. Let $C= B \cup \{e_{j,k}| j \in \Cl(V), k \in [b] \}$. We have $\text{span}(B)= \text{span}(V)$ and $\text{span}(C)= \text{span}(W)$. So, it suffices to show that $\Clh(B)= \Clh(C)$. \newline

Let $M_{in} \in \F^{k \times nb} $ be the matrix where the rows of $B_{in}$ are stacked. \newline

We shall use the following claim (proof is deferred after this proof). \newline

\begin{claim} \label{claim: closure_is_full}
    There exists a index $\text{col-closure}(i) \in [b]$ for each $i \in \Cl(B)$ such that the columns $(i, \text{col-closure}(i))$ of $M_{in}$ are linearly independent.
\end{claim}

We shall to show that any set of blocks which is acceptable for $C$ is also acceptable for $B$ - this implies $\Clh(B)=\Clh(C)$. Note that the matrix of $B$ looks like the following.
$$
\begin{bmatrix} 
0 & M_{in} & 0 \\
* & * & * \\
* & * & * \\
* & * & * \\
\end{bmatrix}$$
where the entries marked $*$ are arbitrary. The matrix of $C$ looks like the following:

$$\begin{bmatrix} 
0 & M_{in} & 0 \\
* & * & * \\
* & * & * \\
* & * & * \\
0 & I & 0
\end{bmatrix}
$$
Let $S \subseteq [n]$ be a $C$-acceptable set of blocks. For each $i \in S$, there is an index $\text{col}(i)$ such that the corresponding columns in $C$ are linearly independent. We need to find another set of indices $\text{col-new}(i)$ for each $i \in S$ such that the corresponding columns in $B$ are linearly independent. We choose these indices as follows:
\begin{align*}
\text{col-new}(i)= \begin{cases} \text{col}(i) & \text{ if }i \not \in \Cl(B) \\
                                 \text{col-closure}(i) & \text{ if } i \in \Cl(B) 
\end{cases}
\end{align*}
It is easy to see this choice works. Suppose some linear combination of these columns were 0. This linear combination cannot include any $i \in \Cl(B)$ - since the corresponding columns in $M_{in}$ are linearly independent. Thus, this linear combination only includes columns not in $\Cl(B)$ - however, for any $i \in \Cl(B)$, the columns in $B$ and $C$ are identical (upto a fixed number of trailing 0's).
\end{proof}

It remains to prove Claim \ref{claim: closure_is_full}.

\begin{proof} (of Claim \ref{claim: closure_is_full}).

We inductively prove the following statement.

\begin{mdframed}
    Let $V$ be a set of vectors with $\Cl(V)= S$. Let $V_{in}= \text{span} \{ v \in V, \text{supp}(v) \in V\}$. There exists a choice of index $\text{col-clos}(i)$ for each $i \in \Cl(V)$ such that the following statement is true.
    
    Let $T= \{v_1, v_2, \cdots, v_k \}\in V_{in}$ be a set of vectors that span $V_{in}$. Let $M_T$ be the matrix where the rows of $T$ are stacked on top of one other:
    $$
    M_T= \begin{bmatrix}
        \hline v_1 \\ \hline \\ v_2 \\ \hline \cdots \\ \hline \cdots \\ 
        \hline \\
        v_k \\ \hline
    \end{bmatrix}
    $$
    Then, the columns $(i, \text{col}(i))$ in $M_T$ are linearly independent.
\end{mdframed}

Note that this statement is true for one set of vectors that span $V_{in}$ if and only if it is true for all sets of vectors that span $V_{in}$ (the statement is saying $\Clh(V_{in})= \Cl(V)$, and amortized closure and closure depend only on linear span, Fact \ref{fact: depends only on span}). \newline

We prove this statement by induction. The base case when $\Cl(V)=\phi$ is trivial. Consider the case when we a new vector $u$ to $V$ and the closure changes from $S_{old}$ to $S_{new}$.  Let $V_{extra} = \text{span} \{ v \in V \setminus V_{in}, \text{supp}(v) \subseteq S_{new}\}$.Let $w_1, w_2, \cdots , w_l$ be a set of vectors in $V \cup \{u\}$ which span $V_{extra}$. We have $\text{span} \{ v \in V \cup \{u\} | \text{supp}(v) \subseteq S_{new}\} = \text{span}(v_1, v_2, \cdots , v_k, w_1, w_2, \cdots , w_l)$. 

Let

$$M_{old}= \begin{bmatrix}
    \hline v_1 \\
    v_2 \\
    \cdots \\
    \cdots \\
    v_k \\
    \hline
    
\end{bmatrix}$$

$$
M_{extra}= \begin{bmatrix}
    \hline w_1 \\
    w_2 \\
    \cdots \\
    \cdots \\
    w_l \\
    \hline
\end{bmatrix}$$

$$
M_{new}= \begin{bmatrix}
\hline \\
    M_{old}  & 0 & 0 \\
    \\
    \hline
    \\
    & M_{extra} \\
    \\\hline
    
\end{bmatrix}$$

By inductive hypothesis, we have a choice $\text{choice}(i) \in [b]$ for each $i \in S_{old}$ such that the columns $(i, \text{choice}(i))$ of $M_{old}$ are linearly independent.  \newline

We need to find a set of linearly independent columns of $M_{new}$ from distinct blocks. \newline

 Since projecting out the closure keeps the rest of the subspace safe, the set $V_{extra}[\setminus S]$ is a safe set. By definition of extra, for every $i \in S_{new} \setminus S_{old}$, there is a choice of index $\text{ind}(i) \in [b]$ such that the columns $(i, \text{ind}(i))$ of $M_{extra}$ are linearly independent. To complete the induction step, we need to exhibit a choice $\text{new-choice}(i) \in [b]$ for each $i \in S_{new}$ so that the columns $(i, \text{new-choice}(i)) $ in $M_{new}$ are linearly independent. We do this as follows:

 \begin{align*}
     \text{new-choice}(i)= \begin{cases}
         \text{choice}(i) & \text{ if }i\in S_{old} \\
         \text{ind}(i) & \text{ if }i \in S_{new} \setminus S_{old}
         \end{cases}
 \end{align*}

It is easy to see the corresponding columns are linearly independent. If any linear combination is zero, that cannot include any columns from $S_{old}$ - because the corresponding columns of $M_{old}$ are linearly independent, and the newly added columns are 0 on the first $k$ entries. So the linear combination can only include columns from $S_{new} \setminus S_{old} $- and this set of columns is linearly independent by construction of $\text{ind}$.
    
\end{proof}

    \begin{proof} \textbf{(of Corollary \ref{both_are_nice})}
    By Lemma \ref{lem: amortized_closure_continuity} we have that $\Cl(A)= \Cl(B)$, so $A_y, B_y$ are nice affine spaces. It remains to show that $\text{codim}(B_y)= \text{codim}(A_y)+1$.

    Let $A= \{x | Mx=b\}$ and let the set of rows of $M$ be $v_1, v_2, \cdots , v_k$. Let $B= \{x | \tilde{M}x=\tilde{b}\}$ where $\tilde{M}$ has $k+1$ rows, the first $k$ of which are $v_1, v_2, \cdots , v_k$. Let the last row be $w$. \newline

    Let $S= \Cl(A)=\Cl(B)$. The set of defining linear forms of $A_y$ is $v_1[\setminus S], v_2[\setminus S], \cdots ,  v_k[\setminus S] $ and the set of defining linear forms of $B_y$ is $v_1[\setminus S], \cdots , v_k[\setminus S], w[\setminus S]$. We wish to show $w[\setminus S]$ is linearly independent from $v_1[\setminus S], v_2[\setminus S], \cdots , v_k[\setminus S]$. This is equivalent to showing that $w$ does not lie in $\text{span}(\{v_1, v_2, \cdots , v_k\} \cup \{e_{i,j} | i \in S, j \in [b] \})$. FTSOC assume $w$ lies in $\text{span}(\{v_1, v_2, \cdots , v_k\} \cup \{e_{i,j} | i \in S, j \in [b]\})$. Thus, $w= r+s$ for some $r \in \text{span}(v_1, v_2, \cdots , v_k)$ and $s \in \F^n$ such that $\text{supp}(s) \subseteq S$. Since $\Cl$ and $\Clh$ of a set depend only on its linear span, we can WLOG replace $w$ by $s$. Hence, assume $\text{supp}(w) \subseteq S$. \newline

    By Lemma \ref{closure remains same}, we have that
    $$ \Clh (\{v_1, v_2, \cdots , v_k, w \}) \subseteq \Clh (\{v_1, v_2, \cdots , v_k\} \cup \{e_{i,j} | i \in S, j \in [b]\}) = \Clh(\{v_1, v_2, \cdots , v_k\})$$
    This is a contradiction, since we assumed that
    $$ |\Clh(\{v_1, v_2, \cdots, v_k, w)|= |\Clh(v_1, v_2, \cdots , v_k)|+1$$
        
    \end{proof}